\documentclass[12pt, letterpaper]{article}
\usepackage{amsmath,amssymb,mathrsfs,amsthm}
\usepackage{microtype}
\usepackage[T1]{fontenc}

\usepackage{times}

\usepackage{setspace}

\usepackage{xcolor}
\definecolor{BLUE}{RGB}{0,0,100}
\usepackage[margin=1in]{geometry}

\usepackage[colorlinks=true,linkcolor=BLUE,citecolor=BLUE,urlcolor=BLUE]{hyperref}
\usepackage{comment}
\usepackage{graphicx}
\usepackage{enumitem}
\usepackage{fnpct}
\usepackage{subcaption}
\usepackage{booktabs}
\usepackage{multirow}
\newenvironment{proofsketch}{\begin{proof}[Proof sketch]}{\end{proof}}

\usepackage[title,titletoc]{appendix}

\theoremstyle{plain}
\newtheorem{theorem}{Theorem}%
\newtheorem{proposition}[theorem]{Proposition}
\newtheorem{corollary}{Corollary}[theorem]

\newtheorem{lemma}{Lemma}
\newtheorem*{corollary*}{Corollary}

\theoremstyle{definition}
\newtheorem{observation}{Observation}%

\newtheorem{remark}{Remark}%

\newtheorem*{remark*}{Remark}
\newtheorem{example}{Example}%
\newtheorem*{claim*}{Claim}

\usepackage{natbib}
\DeclareMathOperator{\supp}{supp}
\DeclareMathOperator{\E}{\mathbf{E}}

\DeclareMathOperator*{\argmax}{arg\,max}

\DeclareMathOperator{\conv}{conv}
\DeclareMathOperator{\ext}{ext}

\renewcommand{\d}{\mathop{}\!\mathrm{d}}

\newcommand{\lbar}{\underline}
\newcommand{\reals}{\mathbb R}

\DeclareSymbolFont{operators}   {OT1}{lmr} {m}{n}
\DeclareMathOperator{\MPS}{\mathrm{MPS}}
\newcommand{\MPSw}[1]{\MPS(#1 \cdot \mathbf{1}_{[\theta_0,\bar\theta]})}
\newtheorem*{example*}{Example}

\title{Allocating Positional Goods: A Mechanism Design Approach%
\thanks{I thank Krishna Dasaratha, Alex Gershkov, Jonathan Libgober, Bart Lipman, Chiara Margaria, Teddy Mekonnen, Juan Ortner, Marcel Preuss, and participants at EC'25, IIOC 2026, Midwest Theory Conference (Ohio State), and 36th Stony Brook Game Theory Conference for their helpful comments.
}
}
\author{
    Peiran Xiao\thanks{Department of Economics, University of Southern California. Email: \href{mailto:peiran.xiao@usc.edu}{{peiran.xiao@usc.edu}}.}
}
\date{\today %
}

\begin{document}
\maketitle

\onehalfspacing

\begin{abstract}
    I study the optimal allocation of positional goods, where consumers' concern for relative consumption creates externalities. Applications include luxury goods, priority services, education, and organizational hierarchies. Using a mechanism design approach, I characterize feasible allocations through a majorization condition. Under Myerson regularity, the revenue-maximizing mechanism fully separates participating buyers, with possible exclusion at the bottom. Selling a single level guarantees at least half the maximum revenue.
    When all buyers are served, restricting the seller to a single level increases consumer surplus under an increasing failure rate (IFR).
    When the seller is restricted to a single level, expanding coverage also benefits consumers under IFR but may harm them otherwise.
    I also characterize the welfare-maximizing mechanism with and without subsidies.
\end{abstract}


\clearpage

\begin{quote}
    ``\emph{If everyone stands on tiptoe, no one sees better.}''

    \hfill  --- Fred Hirsch (1976), \emph{Social Limits to Growth}
\end{quote}

\section{Introduction}
Positional goods are goods or services whose value to consumers depends on their relative position in consumption.
\cite{Hirsch1976} introduced the concept of positional goods, describing them as either ``scarce in some absolute or socially imposed sense'' or ``subject to congestion or crowding through more extensive use.''
The allocation of such goods generates \emph{externalities}:
moving one consumer up necessarily pushes another down, and as more
consumers buy the same good, its value to each diminishes.

Examples of positional goods abound.
Many consumer goods, such as jewelry, luxury cars, and non-fungible tokens (NFTs), are positional because individuals derive utility from social comparisons \citep{Frank1985a, CarlssonJohansson-StenmanMartinsson2007, LundyRamanKominers2025}.
These goods are also called status goods \citep{CharlesHurstRoussanov2009, Rayo2013, BursztynFermanFiorin2018}, as their value stems from the status they confer relative to others.
Other goods are positional due to capacity constraints.
An example is priority services \citep{GershkovWinter2023}, such as priority boarding and priority rideshare pick-up.
Consumers purchase these services to reduce waiting times by moving ahead of others, but their value diminishes as more consumers purchase them; with multiple priority
levels, the value of each level depends on relative consumption as well.
More broadly, education can be viewed as a positional good, as students compete for scarce college seats and job opportunities \citep{Durst2021,KimTertiltYum2024,KrishnaLychaginOlszewski2026}.
\footnote{
Education is the classic example in \cite{Hirsch1976}: ``The value to me of my education depends not only on how much I have but also on how much the man ahead of me in the job line has.''
In East Asia, the education system is viewed as positional \citep{Economist2021} and ``a zero-sum game'' \citep{JiaLiCousineau2025}.
}

Positional goods differ fundamentally from ordinary goods. 
For ordinary goods, consumers derive utility from the good's intrinsic quality; for positional goods, consumers' utility depends on others' consumption.
This distinction gives rise to different optimal mechanisms.
When selling ordinary goods of the same intrinsic quality, the optimal mechanism is a posted price: the seller announces a single price, and each buyer decides whether to purchase. A buyer who purchases the good obtains its value regardless of how many others do so.
By contrast, for positional goods, the value of the good diminishes as more buyers purchase it: a good meant to confer the highest status ceases to do so once sold to many buyers.
Therefore, selling them at a single posted price is no longer optimal, even when they are homogeneous in intrinsic quality. Instead, by branding them into different tiers, the seller creates differences in status value that are endogenously determined in equilibrium, thereby enabling price discrimination.
Indeed, positional goods in practice are commonly offered in multiple tiers. %

The provision of multiple tiers of positional goods is largely driven by externalities.
This logic differs from that in standard second-degree price discrimination for ordinary goods \citep{MussaRosen1978}, where the seller offers products of varying intrinsic quality to screen consumers.
There, product differentiation hinges on a positive marginal cost of producing a higher-quality good: when the marginal cost is zero, the optimal mechanism collapses to a posted price.
\footnote{
This assumes that the buyer's utility is multiplicatively separable and that quality is bounded above.
When the buyer's utility is linear in quality, product differentiation further requires a strictly increasing marginal cost: with a constant marginal cost, the optimal mechanism collapses to a posted price.
}
For positional goods, by contrast, the marginal cost of producing a higher tier can be virtually zero because intrinsic quality hardly differs across tiers, and
yet offering multiple tiers remains optimal.
\footnote{
For pure positional goods, such as credit cards \citep{BursztynFermanFiorin2018}, NFTs, digital collectibles, and priority boarding, tiers differ only in the relative position they confer and cost virtually the same to produce.
Even for luxury goods where quality does vary across tiers, the price premium is predominantly driven by status concerns \citep{BagwellBernheim1996,ChaoSchor1998,KapfererValette-Florence2021}.
}

Positional goods also raise distinct welfare considerations. Unlike ordinary goods, which can in principle be allocated to all consumers, a positional good cannot be allocated to all without diminishing its value: “what each of us can achieve, all cannot” \citep{Hirsch1976}. As a result, positional goods remain scarce and competitive despite economic growth, which raises a natural question: how should such goods be allocated from a welfare perspective?
The answer is nontrivial due to positional externalities.
Dampening competition may benefit consumers by sparing them from a zero-sum status game, but it may also harm those who value status highly. Expanding access to positional goods creates a similarly distinctive trade-off, as it may dilute the status of existing consumers.

In this paper, I study the optimal allocation of positional goods using a mechanism design approach. %
I assume that buyers care about their relative positions, or \emph{status}, defined by the mass of consumers who purchase goods at lower levels (or opt out) plus half the mass of consumers at the same level.
Buyers privately know their valuations of status (i.e., types).
The seller offers a menu with one or more levels of positional goods and sets a price for each level. Each buyer can purchase a good to obtain its intrinsic value and the status it confers, or opt out and pay nothing.
Using a mechanism design approach, I study direct mechanisms that
specify a status profile and a payment schedule as functions of buyers'
(reported) types.

I first characterize which status profiles are feasible, in the sense that
they can arise from some allocation of positional goods.
Positional externalities impose a nontrivial feasibility constraint: for example, it is impossible to assign every consumer the highest status at its nominal value.
Building on the theory of extreme points and majorization \citep[henceforth KMS]{KleinerMoldovanuStrack2021}, I show that feasibility admits a simple characterization in quantile space. After transforming types into quantile ranks, an incentive-compatible status profile is feasible if and only if it is majorized by (i.e., a mean-preserving spread of) the identity function on the participating interval.
Intuitively, if each participating buyer receives a
distinct tier, they are fully separated, and each participant's status is their quantile rank. At the other extreme, if all participants are pooled into a single tier, their status is the average quantile rank of the participating interval.
In general, pooling buyers into the same tier coarsens the status profile while preserving average status, making it a mean-preserving spread of the identity function.

Then, I turn to revenue maximization. If the type distribution is regular in the sense of \cite{Myerson1981}, the seller's revenue increases as she offers more levels of positional goods.
Therefore, the revenue-maximizing mechanism fully separates participating buyers, with possible exclusion at the bottom.
This allocation can be implemented by an all-pay auction with a reserve price, where the buyer who pays more obtains a higher position.
If the type distribution violates regularity, the optimal mechanism may pool some types into the same level.

Although a single-tier posted-price mechanism is never optimal for the seller, it guarantees at least half the maximum revenue.
This result is useful when the number of tiers is subject to practical considerations.
The argument is particularly simple when the positional good has no intrinsic value.  
Consider an auxiliary problem of selling an indivisible ordinary good.
In this problem, the optimal mechanism is a posted price that allocates the good with probability one to every buyer who purchases it.
By contrast, when the good is positional, the seller cannot assign the highest status of $1$ to every participating buyer without diluting it.
Thus, the maximum revenue from selling the positional good is (strictly) below that in the auxiliary problem. 
Nevertheless, by selling a single positional good, the seller can assign every participating buyer a status of at least $1/2$. Charging half the corresponding price from the auxiliary problem therefore generates at least half of the auxiliary problem’s maximum revenue, and hence at least half of the maximum revenue attainable from selling positional goods.

I also analyze the welfare effects of regulations on the number of tiers and service coverage. The directions of these effects are \emph{a priori} ambiguous.
When more tiers are allowed, the seller may offer a finer status profile, which generates efficiency gains by assigning higher status to consumers who value it more, but it also allows the seller to extract more surplus through higher prices.
I show that when the seller optimally chooses to serve all consumers, if the type distribution has an increasing failure rate (IFR), consumer surplus decreases as the seller offers a finer status profile.
Consequently, restricting the seller to a single tier increases consumer surplus.
The reason is that under IFR, the seller extracts more surplus than a finer status generates, leaving consumers worse off.
If, instead, the distribution has a decreasing failure rate (DFR), a finer status profile increases consumer surplus. Under DFR, superstars at the top benefit from being separated from lower types. Although low types still lose, the heavy tail places sufficient weight on these superstars so that their gains dominate.

Similarly, expanding coverage by lowering the price of the lowest tier has mixed effects: it benefits consumers on the extensive margin but can do harm on the intensive margin, as serving more customers may reduce the status of existing ones.
When the number of levels is unconstrained, expanding coverage increases consumer surplus under Myerson regularity, because existing participants are not pooled with new entrants and thus do not lose status.
However, when the seller is restricted to a single tier, expanding coverage may decrease or increase consumer surplus.
Under IFR, the gain to new participants outweighs the loss to existing ones, so consumer surplus rises.

Next, I study optimal mechanisms that maximize consumer surplus.
Unlike ordinary goods, a positional good cannot be allocated to all consumers without diminishing its value, so the optimal mechanism may involve exclusion and multiple tiers.
When negative transfers are allowed and subject to budget balance, consumer surplus is maximized by full separation with cross-subsidization and without exclusion.
Intuitively, assortative matching is efficient, and a fixed subsidy can redistribute efficiency gains among consumers while maintaining incentive compatibility \citep[see also][]{GershkovSchweinzer2010}.
This mechanism can be implemented by an all-pay auction with a lump-sum subsidy financed by the auction payments, where the buyer who pays more obtains a higher position, and every buyer receives the same subsidy.

When subsidies are unavailable, total pooling without exclusion maximizes
consumer surplus if the type distribution satisfies IFR\@. If instead
the distribution satisfies DFR (e.g., Pareto), consumer surplus under
the nonnegative-price constraint is maximized by full separation without
exclusion. Intuitively, separation harms agents by requiring them to pay
more, but high types may benefit because they value status more,
particularly when the distribution has a heavy tail. In both cases, the
consumer-optimal mechanism involves no exclusion. Unlike with ordinary
goods, exclusion can raise the status of existing consumers who would
otherwise be pooled with the excluded, but its harm to excluded
consumers always exceeds its potential benefit.

Finally, I study social welfare, defined as a weighted sum of the
seller's revenue and consumer surplus. In the application to education,
the welfare weight on revenue captures the productivity of education:
the higher the weight, the more productive education is relative to pure
signaling. I characterize conditions under which full separation with
potential exclusion maximizes social welfare.

The model can also be interpreted as an organizational design problem in which a principal designs the number and size of status categories \citep{MoldovanuSelaShi2007}, and agents who care about their status exert effort to climb the hierarchy or opt out.
Under this interpretation, the
``price'' is effort, which is nonnegative and entails a linear cost, and
an agent's type is their ability, which determines the marginal cost of
effort.
If the ability distribution satisfies Myerson regularity, the effort-maximizing mechanism excludes the lowest types and fully separates all participating agents. 
This resembles the ``rank-and-yank'' system, which ranks employees and terminates underperformers.
In terms of agent welfare, the results imply that an egalitarian organization is optimal under IFR, while a strict hierarchy is optimal under DFR.

The results also have implications for education. %
The public has called for collective efforts to curb educational arms races, in which students invest enormous time, effort, and money to outperform their peers \citep{Economist2021}.
My result supports this view when the ability distribution has a thin tail, in which case coarser performance rankings or lottery-based school assignment can improve student welfare on average.
\footnote{
\cite{KrishnaLychaginOlszewski2026} find that pooling a large fraction of the lowest-performing students leads to a Pareto improvement in college admissions.
}
Conversely, when there are a few superstars in the upper tail, meritocracy can be optimal.
While meritocracy creates a rat race that forces students to exert effort, superstars can benefit from it because they have lower marginal costs and are sufficiently spread out in the tail.
Although low abilities still suffer, the heavy tail also makes the gain to higher abilities dominate in the aggregate.
Thus, when the ability distribution has a thin (heavy) tail, softer competition benefits (harms) students on average.
This dichotomy offers one possible explanation for why admissions at higher levels of education, such as graduate schools, are more meritocratic, while those at lower levels tend to be lottery-based.

As an extension, I consider alternative specifications of status, including convex or concave specifications, varying concerns about same-tier consumption, and signaling-based formulations.
The revenue maximization result also extends to the setting where different tiers of positional goods differ in intrinsic quality (à la \citet{MussaRosen1978}). 
In addition, motivated by the priority services application, I allow the seller to
delay service and use excessive waiting time as an ordeal mechanism to screen consumers.
Instead of excluding types with negative virtual valuations, the revenue-maximizing mechanism serves them with excessive waiting time. 
Finally, I study a variant of
the model where higher types obtain \emph{lower} net utility from
purchasing the good so that agents have incentives to overreport.

\paragraph{Literature Review.}
The idea that consumers' utility depends on comparisons with others' consumption dates back at least to \cite{Veblen1899} \citep[see also][]{Duesenberry1949, Leibenstein1950, BagwellBernheim1996}.
The theoretical literature on positional goods focuses on consumer choice in purchasing positional goods
and the welfare effects of income distribution and government policies \citep{Frank1985, Frank2005, Frank2008,Robson1992,HopkinsKornienko2004}.
\footnote{
    Another strand of literature uses the formulation that stems from \cite{Duesenberry1949} and \cite{Pollak1976} (see the survey by \cite{Truyts2010}), which assumes the consumer's utility depends on others' consumption in addition to their absolute consumption.
    Empirical evidence of positional goods and externalities includes \cite{Luttmer2005}, \cite{AlpizarCarlssonJohansson-Stenman2005}, \cite{CarlssonJohansson-StenmanMartinsson2007}, \cite{BursztynFermanFiorin2018}.
}
This paper complements the literature by studying the supply side: the monopolist provision of positional goods.

The monopolist provision of nonpositional goods has been studied extensively.
\cite{MussaRosen1978} establish the optimality of price discrimination
through product differentiation, but the optimal mechanism degenerates
to a posted price when the marginal cost is constant (and quality is bounded above).
\footnote{With multiplicatively separable utility, a posted price is optimal when the marginal cost is zero.}
More generally, without multiplicative separable utility, \cite{AndersonDana2009} show that product differentiation can be optimal when the marginal cost is nonpositive
(see also \cite{DeneckereMcAfee1996}).
This paper differs from both frameworks in two respects: the value of positional goods is
determined endogenously in equilibrium, and product differentiation is
optimal even with zero marginal cost and multiplicatively separable utility.

\citet[henceforth MSS]{MoldovanuSelaShi2007} study the optimal number
and size of status categories in an organization to maximize agents'
effort when agents care about their relative position. They find that
full separation is optimal if the ability distribution has an increasing
failure rate (IFR).
Methodologically, I use a mechanism design
approach that characterizes feasibility via majorization, while they
use optimal contest design with a finite number of agents. Conceptually,
my framework generalizes theirs by allowing for \emph{exclusion} of
agents and stochastic allocations, while also analyzing agents'
welfare.\footnote{Their definition of status is equivalent, up to an
affine transformation, to mine (see Footnote~\ref{ftn:MSS}).}
Furthermore, in cases without exclusion, I show that IFR is unnecessary
for full separation to maximize effort; in fact, full separation maximizes both
effort and welfare if the distribution satisfies Myerson's regularity
and has a \emph{decreasing} failure rate.

Relatedly, \cite{Rayo2013} studies monopolist provision of conspicuous goods and characterizes the optimal allocation.
Although his definition of conspicuous goods is different in that consumers care about the average type in their assigned category rather than relative positions, his results are qualitatively similar to those of MSS.
In Section~\ref{sec:signaling}, I show an analogous feasibility condition for his specification.
\footnote{
    More generally, \cite{Board2009} allows the value of belonging to a category to take a general form and shows that, under a regularity condition, the revenue-maximizing mechanism segregates agents too finely and excludes too many agents compared to the socially efficient benchmark. 
}

Similar to MSS, \cite{ImmorlicaStoddardSyrgkanis2015} study the optimal design of badges to incentivize contribution via status rewards. 
The most important difference is that they assume status is decreasing in the mass of agents who are \emph{weakly} superior, including those at the same level.
Under this specification, pooling assigns agents in the pool the status of the lowest type of that pool, thereby lowering status for all other pooled agents relative to full separation.
Consequently, pooling is never optimal for agent welfare, and opting out is equivalent to pooling at the bottom tier (see Section~\ref{sec:extensions:gamma}).

In a different vein, \citet{GershkovWinter2023} study the welfare implications of introducing priority service.
Because priority service is a positional good generated by capacity constraints, my analysis extends theirs to multiple, potentially stochastic priority levels while allowing for exclusion using a mechanism design framework.

Another closely related paper is by \cite{LoertscherMuir2022}, who study revenue maximization when selling a fixed inventory of goods with heterogeneous qualities.
Under the linear status specification and without intrinsic value, my model is analogous to a continuous version of theirs in which each quality is in unit supply,
so that positional externalities arise from limited supply.
In particular, pooling multiple types at the same status level corresponds to randomization (or conflation) over qualities in their model.
\footnote{
Under nonlinear status specifications or when consumers place different weights on others at the same level (see Section~\ref{sec:extensions}), this analogy no longer holds.
}
My framework, however, also applies to settings in which positional externalities arise from social comparisons rather than physical scarcity, thereby yielding new insights for applications such as luxury goods.

More broadly, this paper contributes to the literature on mechanism
design with allocative externalities
\citep[e.g.,][]{JehielMoldovanuStacchetti1996,JehielMoldovanu2006,OstrizekSartori2023,AkbarpourDworczakKominers2024,DworczakReuterKominers2026}.
The model features a continuum of agents who sort endogenously into
tiers that induce different status levels; positional externalities
arise from agents' relative consumption rather than from average
characteristics of others \citep{DworczakReuterKominers2026} or network
effects \citep{Csorba2008,OstrizekSartori2023,MeisnerPillath2025}.

Finally, this paper relates to the literature on signaling and matching
tournaments
\citep[e.g.,][]{McAfee2002,HoppeMoldovanuSela2009,Hopkins2023,KrishnaLychaginOlszewski2026}.
This literature studies how agents use costly signals or investments to
improve their rank or matching outcomes. In contrast, this paper studies how a
designer allocates positional goods to induce agents' status.

\section{Model}
\subsection{Setup}
A monopolist seller (she) sells positional goods ${x}\in X\subseteq \reals_{++}$ to a continuum of buyers (he) with unit mass.
A larger ${x}$ represents a higher-level good, such as a higher-end luxury good or a higher-priority boarding group, and the set $X$ is determined by the seller. 
\footnote{
    In general, it is without loss of generality to take $X=\reals_{++}$.
    If the number of levels is constrained, then $X$ is a finite subset of $\reals_{++}$.
}
Let ${x}=0 \notin X$ denote the outside option of not buying.
Because consumers care about their relative position, consumption of positional goods confers \emph{status} (defined later).
Buyers are heterogeneous in their valuations of status (i.e., types), denoted by $\theta$, which has a distribution $F(\theta)$ with continuous density $f(\theta)>0$ on $\Theta = [0,\bar \theta]$.%
\footnote{
I allow $\bar\theta=\infty$ under the assumption that $\E[\theta]<\infty$.
The lower bound can also be generalized to $\lbar\theta>0$, which leads to a complication in consumer surplus maximization under DFR (see Remark~\ref{complications}).
}
If a buyer with type $\theta$ purchases a good at the price $p$, he obtains a payoff of $u(p,s,\theta) = \theta s - p + v(\theta)$, where $s\in[0,1]$ is the status conferred by the good and $v(\theta)\geq0$ is the intrinsic value of the good. 
Assume $v'(\theta)\geq0$ and $v''(\theta)\leq0$, so that higher types derive greater intrinsic value from the good at a diminishing marginal rate.
If the buyer does not purchase (i.e., $x=0$), his status is defined as $s=0$, and he does not obtain the intrinsic value, so his payoff is zero.

Equivalently, $-v(\theta)$ can be interpreted as a type-dependent outside option, and the assumptions imply that higher types suffer more from not buying the good.

\paragraph{Definition of status.}
Let $G$ denote the distribution of buyers' consumption over ${x}\in\tilde{X}\equiv X\cup\{0\}$, so $G({x})$ represents the mass of buyers who consume weakly lower-level goods than ${x}$ or opt out.
Following \cite{Robson1992} and \cite{HopkinsKornienko2004}, I define the \emph{status} of a buyer who consumes ${x}\in X$ (i.e., participant) as
\begin{equation} \label{eqn:status}
    S({x}, G(\cdot)) 
    = 
    \begin{cases}
        \dfrac{G^{-}({x}) + G({x})}{2} \in [0,1], &\text{if } {x}>0 \\
        0, &\text{if } {x}=0
    \end{cases}    
\end{equation}
where $G^-({x}) \equiv \lim_{{x} \to {x}^-} G({x})$ is the mass of buyers who consume \emph{strictly} lower-level goods than ${x}$ (or opt out).
In other words, participants' status equals the mass of buyers who consume strictly lower-level goods plus half the mass of buyers at the same level.%
\footnote{
    \cite{HopkinsKornienko2004} assume a more general specification: $S({x}, G(\cdot)) = \gamma G({x}) + (1-\gamma) G^{-}({x}) + \alpha$, where $\gamma\in [0,1)$ captures the concern about the number of consumers at the same level, and $\alpha\geq0$ is a constant representing a guaranteed minimum status for participants. %
    My model incorporates $\alpha$ through the intrinsic value $v(\theta)$, and many results are robust to general $\gamma$ (see Section~\ref{sec:extensions:gamma}).
    }%
\footnote{
    Without exclusion, this specification is equivalent to the one used in MSS and \cite{DubeyGeanakoplos2010} (up to an affine transformation):
    \(
        S({x}, G(\cdot)) 
        = G^-({x}) - (1-G({x})) \in [-1,1].
    \)
    \label{ftn:MSS}
}
Non-participants ($x = 0$) obtain zero status and contribute to the status of every participant because they are ranked strictly below participants.

Because buyers only care about the induced status, 
the seller essentially sells status $s\in [0,1]$, with distribution $\Psi$, subject to the feasibility constraint
\begin{equation*}\label{eqn:status2}
    s %
      =  
        \begin{cases}
            \dfrac{\Psi^{-}({s}) + \Psi({s})}{2} \in [0,1], &\text{if } {s}\in \supp(\Psi) \setminus \{0\}\\
            0, &\text{if } {s} =0.
        \end{cases}
\end{equation*}
For a given mass of participants, their expected total status is constant, regardless of the distribution of status (whether it has mass points due to pooling).
This reflects the zero-sum nature of status among participants.

This specification of status can arise from interpersonal comparisons \citep{Frank1985, Robson1992, HopkinsKornienko2004}, especially for status goods.
Consumers gain utility from outranking those who buy lower-level goods or cannot afford the good, but gain less from comparisons with consumers who purchase the same level.

In addition to psychological reasons, the specification can also arise from physical capacity constraints. The following example illustrates this through priority services in queuing, where status corresponds to the reduction in waiting time.

\begin{example}[Priority Services \citep{GershkovWinter2023}] \label{ex:queuing}
    Consider a unit mass of consumers who arrive simultaneously, and a seller serves one consumer per unit time.
    The seller offers tiered priority services: consumers at higher priority levels are served first, and those at the same level are served in random order.
    A consumer's payoff is $u(p,t,\theta) = \tilde{v}(\theta) - \theta t - p$, where $t \in [0,1]$ is the (expected) waiting time, $\theta$ is the cost of waiting, and $\tilde{v}(\theta)$ is the value of the good.
    Defining $s = 1 - t$ as the status, or priority value, and $v(\theta) = \tilde{v}(\theta) - \theta$ as the net value of the good, this model is equivalent to the baseline model.
    \footnote{
    The assumption $v'(\theta) \geq 0$ implies $\tilde{v}'(\theta) \geq 1$: higher types have higher net value of the good after accounting for waiting costs when served last.
    Section~\ref{suffering} considers an extension in which $v'(\theta) \leq -1$ (equivalently, $\tilde{v}'(\theta) \leq 0$), so that higher types derive lower net value from the good.
    }   
    In particular, excluded consumers contribute to participants' status because participating buyers need not wait for nonparticipants.
\end{example}

Relatedly, education can be viewed as a positional good, with externalities arising from competition for jobs or social positions, because an individual's education level determines their position in the job queue or social hierarchy.
This application also resembles the application to status contests within organizations.

\begin{example}[Status Contests \citep{MoldovanuSelaShi2007}]
    Consider a principal who maximizes agents' effort by designing status categories, where agents care about their relative position  within the organization.
    The positional good ${x}$ represents the status category,
    the price $p \geq 0$ represents the agent's effort, and the type $\theta$ represents the agent's ability (which determines the marginal cost of effort).
    Each agent has a linear effort cost $p/\theta$ and receives a payoff of $s - p/\theta$ upon attaining status $s \in [0,1]$.
    An excluded agent receives a zero payoff.
    After scaling agents' payoffs by $\theta$, this is a special case of the model with $v(\theta) = 0$ and an additional constraint $p \geq 0$.
\footnote{
The type $\theta=0$ can be treated as a limiting boundary type, who has an infinite marginal cost and therefore optimally exerts zero effort.
}
\end{example}

\subsection{A Mechanism Design Approach}
Consider a direct mechanism $(\chi(\theta,\omega), p(\theta), \sigma(\theta))$, consisting of a (potentially stochastic) allocation rule $\chi\colon \Theta\times \Omega \to \reals_{+}$, a payment function $p\colon \Theta\to \reals$, and
a purchase indicator $\sigma\colon \Theta \to \{0,1\}$.
To allow for stochastic allocations, I introduce a random variable $\omega\in \Omega = [0,1]$ to capture all randomness in the mechanism. 
The random variable $\omega$ is drawn uniformly, independent of types, and is common to all buyers.

In the direct mechanism, the buyer reports his type $\theta$.
If $\sigma(\theta)=1$, the buyer pays $p(\theta)$ and receives the good $\chi(\theta,\omega)>0$.
If $\sigma(\theta)=0$, the buyer does not purchase and receives the outside option $0$.
\footnote{
The purchase indicator $\sigma$ is taken to be deterministic, which captures the assumption that the buyer obtains the intrinsic value $v(\theta)$ if and only if he purchases. When $v(\theta)=0$, it is without loss to set $\sigma\equiv 1$ and allow stochastic nonparticipation by incorporating it to the allocation rule as $\chi(\theta,\omega)=0$ and $p(\theta)=0$.
\label{ftn:v=0}
}
By the taxation principle, the direct mechanism can be implemented by a pricing scheme as a function of lotteries over the set of positional goods $X =\operatorname{range}(\chi)$.
To incorporate the purchase decision, define the extended allocation and payment functions as $\chi \cdot \sigma$ and $p \cdot \sigma$. For simplicity, I will use $\chi$ and $p$ to denote them respectively, with the caveat that $\sigma(\theta) = \mathbf1_{\chi(\theta,\omega)>0}$ does not depend on $\omega$.

Given an allocation $\chi$, each realization $\omega\in\Omega$ induces a distribution over consumption levels $x\in X \cup \{0\}$, given by
$G_{\chi}(x\mid\omega) = \int_\Theta \mathbf 1_{\chi(\theta,\omega)\leq x} \d F(\theta)$.
Since buyers' payoffs depend on the allocation $\chi$ only through the status it confers, the allocation $\chi$ can be summarized by its induced status. Thus, it is equivalent to consider a direct mechanism $(s(\theta),p(\theta),\sigma(\theta))$, where 
$s\colon\Theta\to[0,1]$ is the \emph{status profile} induced by $\chi$, defined as
\[
s(\theta) = \E_\omega [S( \chi(\theta,\omega) ,G_\chi(\cdot\mid\omega))].
\]
In particular, if $\sigma(\theta)=0$, then $\chi(\theta,\omega)=0$ for all $\omega\in\Omega$, and thus $s(\theta)=S(0,G_\chi)=0$.

Let $U(\hat\theta|\theta) = \theta s(\hat\theta) - p(\hat\theta) + v(\theta)\sigma(\hat\theta)$ denote the buyer's payoff when he reports $\hat\theta$ while his true type is $\theta$.
Define $U(\theta) = U(\theta|\theta)$.
By the standard envelope argument, incentive compatibility requires that $U(\theta)=\max_{\hat\theta} U(\hat\theta|\theta)$, which implies that $U$ is absolutely continuous with
$U'(\theta) = s(\theta) + v'(\theta) \sigma(\theta)$ and the following lemmas.

\begin{lemma}\label{lemma:IC}
    A direct mechanism $(s(\theta),p(\theta),\sigma(\theta))$ is incentive-compatible and individually rational if and only if 
    \begin{itemize}
    \item there exists $\theta_0\in[0,\bar\theta]$ such that $\sigma(\theta)=1$ and $U(\theta)>0$ for all $\theta>\theta_0$, and  $\sigma(\theta)= 0$ and $U(\theta)=0$ for all $\theta<\theta_0$;
    \item $s(\theta)$ is increasing;
    \item $U(\theta) = U(0) + \int_{\theta_0}^\theta (s(t) + v'(t))  \d {t}$ for all $\theta \in [\theta_0,\bar\theta]$.
    \footnote{In particular, this implies that $U(\theta_0) = U(0) = 0$ if $\theta_0>0$.}
    \end{itemize}
\end{lemma}

An incentive-compatible status profile $s(\theta)$ is not necessarily feasible, as it may be unable to be induced by an allocation $\chi$.
For example, a binary status $s(\theta)= \mathbf{1}_{\theta\geq \theta_0}$ cannot be induced by any allocation $\chi$.
Formally, say $s(\theta)$ is \emph{feasible} if there exists an allocation $\chi\colon \Theta \times \Omega \to \reals_+$ that induces it---i.e., such that $s(\theta)=\E_\omega [S(\chi(\theta,\omega), G_\chi(\cdot\mid\omega))]$.
The following theorem provides the necessary and sufficient condition for an incentive-compatible (and hence increasing) status profile to be feasible.

\begin{theorem}[Feasibility]\label{thm:feasibility}
    An increasing $s(\theta)$ is feasible if and only if there exists $\theta_0\in[0,\bar\theta]$ such that $s(\theta)$ is majorized by $F(\theta)$ in quantile space on $[\theta_0,\bar\theta]$, that is,  
    \begin{align}
        \int_{\theta}^{\bar\theta}   s(\tilde \theta)  \d F(\tilde \theta) \leq \int_{\theta}^{\bar\theta}  F(\tilde \theta)   \d F(\tilde \theta),\quad  \text{ for all } \theta \in [\theta_0, \bar \theta], %
    \end{align}
    with equality at $\theta=\theta_0$, and $s(\theta)=0$ for all $\theta\in[0,\theta_0)$.
\end{theorem}

The proof (in Appendix~\ref{sec:proofs}) is à la Theorem 3 (Border's Theorem) in KMS. 
Let $\MPS(F\cdot\mathbf1_{[\theta_0,\bar\theta]})$ denote the set of increasing functions $s\colon\Theta\to[0,1]$ that satisfy the majorization condition on $[\theta_0,\bar\theta]$.
Let $\MPS_0(F)=\bigcup_{\theta_0\in[0,\bar\theta]} \{s\in \MPS(F\cdot\mathbf1_{[\theta_0,\bar\theta]})\mid s(\theta)=0 \text{ for all } \theta\in[0,\theta_0)\}$ denote the set of increasing feasible status profiles.

\begin{remark}
    For deterministic allocations $\chi$, a feasible $s(\theta)$ must be an extreme point of $\MPS(F\cdot\mathbf1_{[\theta_0,\bar\theta]})$%
    ---i.e., either the majorization constraint or the monotonicity constraint binds on $[\theta_0,\bar\theta]$---and $s(\theta)=0$ for all $\theta\in [0,\theta_0)$. %
    For stochastic allocations $\chi$, a feasible $s(\theta)$ can be a non-extreme point.
    \footnote{
       If the seller can randomly not serve participants, the feasibility condition becomes: $s$ is weakly majorized by $F$ in quantile space, denoted by $s\in \MPS_w(F)$.
       However, the model assumes that types above (below) the cutoff $\theta_0$ always (never) purchase and obtain an intrinsic value $v(\theta)$.
       If $v(\theta)=0$, this assumption has no bite (see also Footnote~\ref{ftn:v=0}).
    }
\end{remark}

\begin{remark} \label{MPS(F)}
    When exclusion is impossible, we have $s\in \MPS(F)$ and $\E[s]=\E[F]=1/2$.
    \footnote{
        Under MSS's specification $S({x}, G) = G^-({x}) - (1 -  G({x}))$ without exclusion, the feasibility condition becomes $s\in \MPS(2F-1)$ with $\E[s]=\E[2F-1]=0$.%
    }
\end{remark}

To provide some intuition, first consider the full separation case where each participating buyer is assigned to a distinct tier.
Each participant's status is just their quantile: $s(\theta)=F(\theta)$. 
Then, when buyers $\theta \in [\lbar\theta_i,\bar\theta_i]$ are pooled into the same tier $x_i>0$, the status of that tier becomes the average quantile of the buyers in the pool because
\begin{equation*}
    s(\theta) = \frac{G^-({x}_i)+G({x}_i)}2 = \frac{F(\lbar\theta_i)+F(\bar\theta_i)}{2}, \quad \text{ for all } \theta\in [\lbar\theta_i,\bar\theta_i].
\end{equation*}
Thus, pooling coarsens the status profile across buyers, while preserving
the aggregate status, making their status a mean-preserving spread of the quantile.
Randomization also preserves this relationship for each realization $\omega$ and expands the set of feasible status profiles to the exact set of mean-preserving spreads through convexification.

\begin{example}[Two tiers] \label{ex:twolevels}
    Suppose there are two levels of positional goods, $(p_L, s_L)$ and $(p_H, s_H)$.
    Assume $p_L = 0$ so that every buyer participates ($\theta_0=0$).
    Let $\theta^*$ denote the type who is indifferent between the two levels, which is determined by $p_H$. 
    Then, the induced status profile is
    \[s(\theta)=  \begin{cases}
        s_L =  {F(\theta^*)}/{2}, &\text{if }\theta < \theta^*,\\
        s_H =  ({1+F(\theta^*)})/{2}, &\text{otherwise}. 
            \end{cases}\] 
    It is straightforward to check that $\E[s] = 1/2$ and that $ \int_{\theta}^{\bar\theta}  s(t)  \d F(t) \leq \int_{\theta}^{\bar\theta} F(t) \d F(t)$ for all $\theta \in [0,\bar \theta]$, with equality at $\theta = 0$, $\theta^*$, and $\bar \theta$.
    \begin{figure}[htb]
        \centering
        \includegraphics[width=0.4\textwidth]{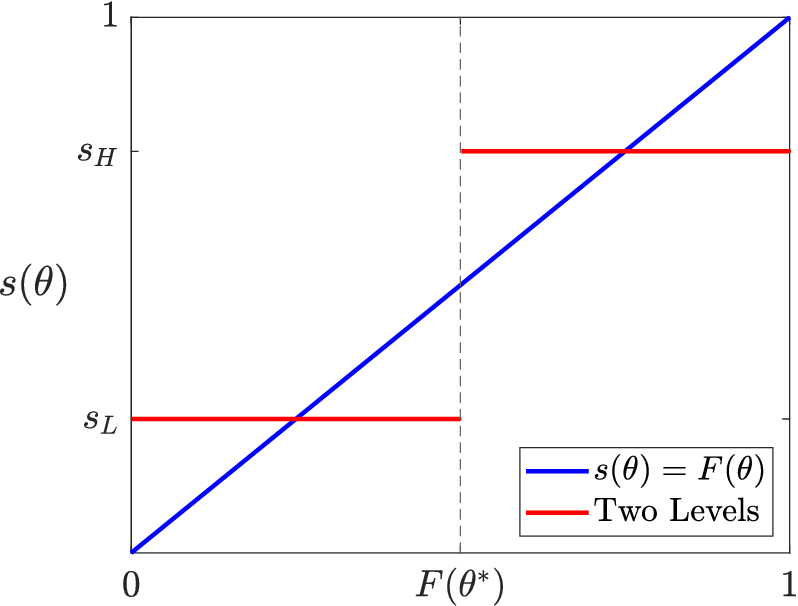}
        \hspace{25pt} \caption{Two Priority Levels}
    \end{figure}
\end{example}

It is worth noting that a binary status profile $s(\theta)=\mathbf{1}_{\theta \geq \theta^*}$, which resembles a posted-price mechanism when selling ordinary goods, is infeasible because it violates the mean-preserving spread condition at the top.
Intuitively, this is because a good meant to confer the highest status $s=1$ ceases to do so once sold to multiple buyers.

Finally, when the seller is constrained to offer at most \(n\geq1\) levels of positional goods, so that \(|X|\le n\), the feasibility condition is more restrictive. Say that a status profile \(s(\theta)\) is \emph{\(n\)-feasible} if there exists an allocation \(\chi\colon \Theta\times\Omega\to X\cup\{0\}\), with \(|X|\le n\), that induces it. For deterministic \(\chi\), any \(n\)-feasible incentive-compatible (and hence increasing) $s(\theta)$ is a step function on \([\theta_0,\bar\theta]\), generated by partitioning participating types into at most \(n\) pooling intervals.

\section{Revenue Maximization}
\label{exclusion}

In this section, I characterize the optimal mechanism that maximizes the seller's revenue and study its welfare implications.
Two extreme allocations are particularly of interest: full separation of participants (i.e., $s(\theta) = F(\theta) \cdot \mathbf{1}_{\theta \geq \theta_0}$ at different prices) and pooling among participants (i.e., $s(\theta) =  \frac{1+F(\theta_0)}{2} \cdot \mathbf{1}_{\theta \geq \theta_0}$ at the same price) for some cutoff $\theta_0 \in [0, \bar \theta)$. %
The former can be induced by offering a continuum of positional goods, and the latter can be induced by offering a single good (tier).

\subsection{Seller's Problem}
\label{sec:revenue}
\allowdisplaybreaks
The revenue maximization problem is given by
\begin{align}
    \max_{s(\theta), p(\theta), \theta_0} \int_{\theta_0}^{\bar \theta} p(\theta) \d F(\theta) 
\end{align}
subject to the following constraints for all $\theta\in [\theta_0,\bar\theta]$:
\begin{align}\label{constraintsw}
    & U(\theta) \equiv \theta s(\theta) - p(\theta) + v(\theta) \geq  0  &&\mbox{(IR)}\\
    & U(\theta) =  U(\theta_0) + \int_{\theta_0}^\theta (s(t) + v'(t)) \d {t}  &&\mbox{(IC)}\\
    & s(\theta) \mbox{ is increasing} \\
    & s \in \MPSw{F} &&\mbox{(MPS)}
    \label{MPSw} 
\end{align}

Define $J(\theta) = \theta - \frac{1-F(\theta)}{f(\theta)}$ and $J_v(\theta) = v(\theta) - \frac{1-F(\theta)}{f(\theta)}v'(\theta)$.
Say $F$ satisfies Myerson's regularity if $J(\theta)$ is increasing.
For expositional simplicity, if $J(\theta)$ is increasing, assume it is strictly single-crossing, that is,
there exists a unique $\theta_0^J\in[0,\bar\theta)$ such that $J(\theta)<0$ for all $\theta<\theta_0^J$ and $J(\theta)>0$ for all $\theta>\theta_0^J$.

By standard arguments, expected revenue is given by
\begin{align*}
    R 
    =  \int_{\theta_0}^{\bar \theta} \left( J(\theta) s(\theta) + J_v(\theta)  - U(\theta_0) \right) \d F(\theta)
\end{align*}
It is optimal to set $U(\theta_0)=0$ by setting $p(\theta_0) = \theta_0 s(\theta_0) + v(\theta_0) \geq 0$.  %
Thus, the revenue maximization problem is equivalent to
\begin{equation}
    \max_{s \in \MPSw{F}, \theta_0} \int_{\theta_0}^{\bar \theta}   \left( J(\theta) s(\theta) + J_v(\theta) \right)   \d F(\theta)  
\end{equation}
For any $\theta_0\in[0,\bar\theta]$, the objective is a
continuous linear functional of $s$ over
$\MPS(F\cdot \mathbf1_{[\theta_0,\bar\theta]})$. By KMS, this set is convex and
compact, so Bauer's maximum principle implies that the maximum is
attained at an extreme point, which is induced by a deterministic allocation.%

For any two increasing status profiles $s,\hat{s}\in\MPS(F\cdot \mathbf1_{[\theta_0,\bar\theta]})$ such that $s(\theta)=\hat s (\theta) = 0$ for all $\theta<\theta_0$, say $s$ is \emph{finer} than $\hat{s}$ ({in terms of majorization}), denoted by $s\succ \hat s$, if 
$\int_{\theta}^{\bar\theta} (s(t) - \hat{s}(t)) \d F(t) \geq 0$ for all $\theta\in[\theta_0,\bar\theta]$.
\footnote{
    If $s$ induces a finer partition of types than $\hat{s}$, then $s$ is \emph{finer} than $\hat{s}$ ({in terms of majorization}); the converse is not true.
}
If $J(\theta)$ is increasing, the Fan--Lorentz Theorem  implies that the seller can increase revenue by offering a mechanism with a finer status profile.
Hence, we have the following proposition.

\begin{proposition}[Revenue Maximization]\label{prop:revmax}
    If $J(\theta)$ is increasing,
    the revenue-maximizing mechanism excludes $\theta< \theta_0^*$ and fully separates $\theta\geq \theta_0^*$:
    \[
        s^*(\theta) =  F(\theta) \cdot \mathbf1_{\theta\geq\theta_0^*},\quad p^*(\theta) =  \left( \theta_0^*F(\theta_0^*)  + \int_{\theta_0^*}^\theta {t} \d F(t) + v(\theta_0^*) \right)  \cdot \mathbf1_{\theta\geq\theta_0^*},
    \]
    where the optimal cutoff type is $\theta^*_0 \in \argmax_{\theta_0}  \int_{\theta_0}^{\bar \theta} J(\theta) F(\theta) \d F(\theta) + v(\theta_0)(1-F(\theta_0))$.

    If $J(\theta)$ is not monotonic, define $\tilde{J}(\theta;\theta_0) = \int_{\theta_0}^\theta J(t) \d F(t)$ and let $K(\tau) = \conv \tilde J(F^{-1}(\tau))$ denote its convex hull on $\tau \in [F(\theta_0),1]$ in quantile space.
    The revenue-maximizing mechanism excludes $\theta<\theta^*_0$, pools types on each ironing interval where $K\circ F < \tilde J$, and separates types if $K\circ F = \tilde J$, where $\theta^*_0 \in \argmax_{\theta_0}  \int_{\theta_0}^{\bar \theta} F(\theta) \d K(F(\theta)) + v(\theta_0)(1-F(\theta_0))$.
\end{proposition}

If the type distribution satisfies Myerson's regularity, the revenue-maximizing mechanism can be implemented by an all-pay auction with a reserve price, where buyers who pay more (and reach the reserve price) receive higher levels of positional goods.
In applications such as education or status within organizations, this may be interpreted as an all-pay contest with a minimum effort requirement.

In general, the optimal mechanism is derived by applying the ironing technique to obtain the convex hull of $\tilde{J}(\theta)$ in quantile space \citep{Myerson1981,Toikka2011}.
If the virtual valuation $J(\theta)$ is decreasing for some types, the revenue-maximizing mechanism pools them into the same positional good level along with some adjacent types.

The results have several implications.
First, under Myerson's regularity condition, it is optimal to offer as many levels of positional goods as possible.
As the number of levels increases, the seller can induce a finer status profile, which not only generates efficiency gains but also allows the seller to extract higher revenue.
This offers an explanation for why airlines introduce many boarding groups, and luxury companies proliferate products from factory stores to high-end exclusives.

Compared to selling ordinary goods with intrinsic value $v(\theta)$ only, attaching status value benefits the seller, which may explain why luxury companies invest heavily in marketing to cultivate status value.

In addition, the seller also benefits from excluding low types.
In particular, when the intrinsic value $v(\theta)=0$, all types with negative virtual valuation are excluded.
The following corollary characterizes properties of the optimal exclusion.

\begin{corollary}[Exclusion] \label{cor:exclusion}
    If $J(\theta)$ is increasing, the optimal $\theta^*_0$ has the following properties.
    \begin{enumerate}[label=(\roman*)]
        \item Intrinsic value reduces exclusion: $\theta^*_0 \leq \theta_0^J$.
        The equality holds if $v(\theta)$ is linear (i.e., $v(\theta) = \alpha\theta$ for some $\alpha \geq 0$).
        
        \item If $v(0) \geq \frac{v'(0)+1}{f(0)}$, then zero exclusion is optimal---i.e., $\theta_0^* = 0$.
        \footnote{
          If \(J(\theta)F(\theta)+J_v(\theta)\) is increasing, then \(v(0)\ge v'(0)/f(0)\) is sufficient.
        }
    \end{enumerate}
\end{corollary}
The intrinsic value $v(\theta)$ reduces the optimal exclusion compared to the zero intrinsic value benchmark,
 where agents with negative virtual valuation $J(\theta)<0$ are excluded.
In addition, if the baseline intrinsic value $v(0)$ is large and there are many low types ($f(0)$ is large), it is optimal for the seller to serve every type.

\subsection{One-tier Approximation}
\label{sec:approx}

In practice, the number of levels may be subject to implementation costs or regulatory constraints.
When there are costs associated with offering multiple levels, a finite number of levels may be optimal.
\footnote{Let $C_{n}\geq0$ denote the cost of offering $n$ levels. If $C_{n}$ is strictly increasing and convex in $n$, the optimal number of levels is finite, since the maximal revenue is bounded above (by $R^*$).}
When the number of levels is constrained, lifting the constraints strictly increases revenue, as shown in Corollary~\ref{cor:sep}.

\begin{corollary}\label{cor:sep}
    The revenue-maximizing mechanism never pools $\bar\theta$ with a positive measure of other types. Moreover, allowing more levels of positional goods \emph{strictly} increases revenue.
\end{corollary}

The result implies that selling a single tier of positional goods at a posted price is never optimal for the seller, in contrast to ordinary goods, as she can always increase revenue by separating the highest type from the rest.
Nevertheless, selling a single tier at the optimal posted price achieves at least half the maximum revenue attainable from selling positional goods.

\begin{proposition}[Approximation] \label{prop:approx}
   The seller can obtain at least half the maximum revenue by selling a single tier.
\end{proposition}

\begin{proofsketch}
    For ease of exposition, assume $v(\theta)=0$, so that the buyer's payoff is $u(p,s) = \theta s - p$.
    Consider an auxiliary problem of selling an indivisible ordinary good, in which $u(p,q) = \theta q - p$ when the buyer receives the good with probability $q\in[0,1]$.
    In the auxiliary problem, the set of incentive-compatible allocations is $\mathcal M = \{q\colon \Theta\to [0,1] \mid  \text{$q$ is increasing}\}$, and
    a posted price $p^* \in \argmax_p p (1-F(p))$ (and $q^*(\theta) = \mathbf{1}_{\theta\geq p^*}$) is optimal. Let $R_{aux}^* = p^*(1-F(p^*))$ denote the maximum revenue.

    Let $R^*$ denote the maximum revenue from selling positional goods.
    When selling positional goods, because the feasible set $\MPS_0(F)\subseteq\mathcal M$, we have $R^* \leq R_{aux}^*$.
    \footnote{
        Moreover, inequality is strict because any optimal mechanism in the auxiliary problem must assign $q=1$ a positive mass of highest types, which is infeasible when selling positional goods.
    }
    Let $R_1$ denote the maximum revenue from selling a single tier of positional goods.
    Selling a single tier at price $\frac{1+F(p^*)}{2}p^*$ excludes buyers with types $\theta < p^*$, as in the auxiliary problem, and guarantees each participant at least $s = \frac{1+F(p^*)}{2} \geq \frac{1}{2}$. The resulting revenue is thus at least $\frac{1}{2} R_{aux}^*$. Hence, we have $R_1 \geq \frac{1}{2} R_{aux}^* > \frac{1}{2} R^*$.
    
    The complete proof, which allows $v(\theta) \geq 0$, is given in Appendix~\ref{sec:proofs}.
\end{proofsketch}

\begin{remark}
    The lower bound does not require IFR or Myerson's regularity. 
\end{remark}

The auxiliary problem of selling indivisible ordinary goods highlights the distinction between selling ordinary goods and positional goods.
In the auxiliary problem, selling an item to one consumer has no externalities on others, so the seller can allocate $q=1$ to multiple buyers. By contrast, when selling positional goods, because consumers care about their relative positions, negative externalities arise---the seller cannot allocate the highest status $s=1$ to multiple buyers.
Thus, the seller of positional goods is subject to an additional feasibility constraint, resulting in a lower revenue.

\paragraph{Graphical Illustration.}
Figure~\ref{fig:approx} plots the \emph{revenue curve} $R(\tau)=(1-\tau)F^{-1}(\tau)$ in quantile space $\tau=F(\theta)$.
For exposition, I focus on the regular case where $J(\theta)$ is increasing, so that the revenue curve is concave, and assume the intrinsic value $v(\theta)=0$.

The \emph{blue} area in Figure~\ref{fig:approx} represents $R^*$, the revenue generated by the revenue-maximizing allocation $s^*(\theta) = F(\theta)\cdot \mathbf{1}_{\theta\geq\theta_0^*}$.
To see this, note that $R'(\tau)=-J(F^{-1}(\tau))$, so substituting $\tau=F(\theta)$ and integrating by parts,
the maximum revenue is
\[\int_{\theta_0^*}^{\bar\theta} J(\theta) F(\theta) \d F(\theta) = - \int_{\tau^*_0}^{1} \tau \d R(\tau) = \tau^*_0 R(\tau^*_0) + \int_{\tau^*_0}^{1} R(\tau)\d\tau,\]
which equals the area of the blue region.

\begin{figure}[hbt]
    \centering
    \begin{subfigure}[h]{0.45\textwidth}
        \centering
        \includegraphics[width= \textwidth]{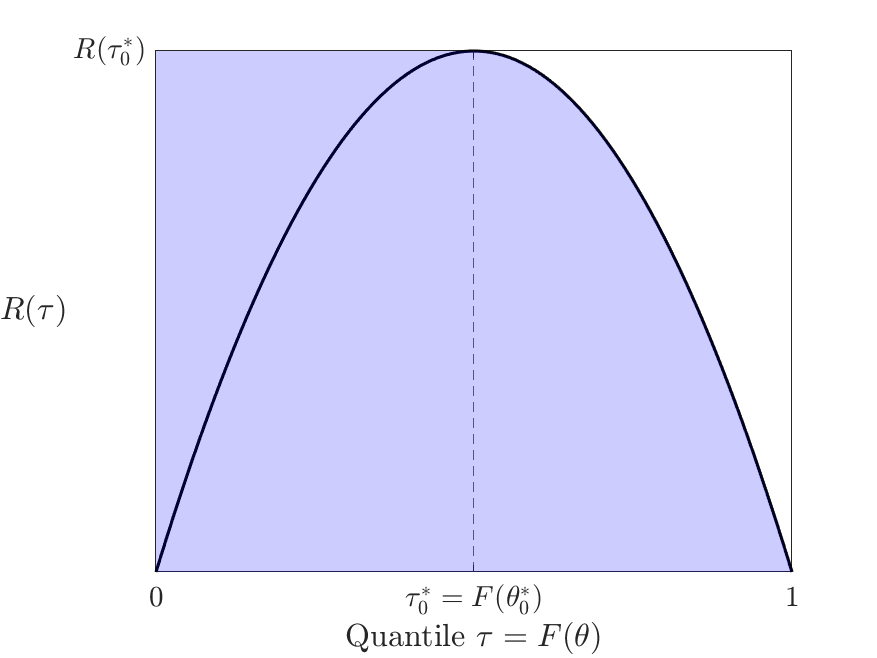}
    \end{subfigure}
    \begin{subfigure}[h]{0.45\textwidth}
        \centering
        \includegraphics[width= \textwidth]{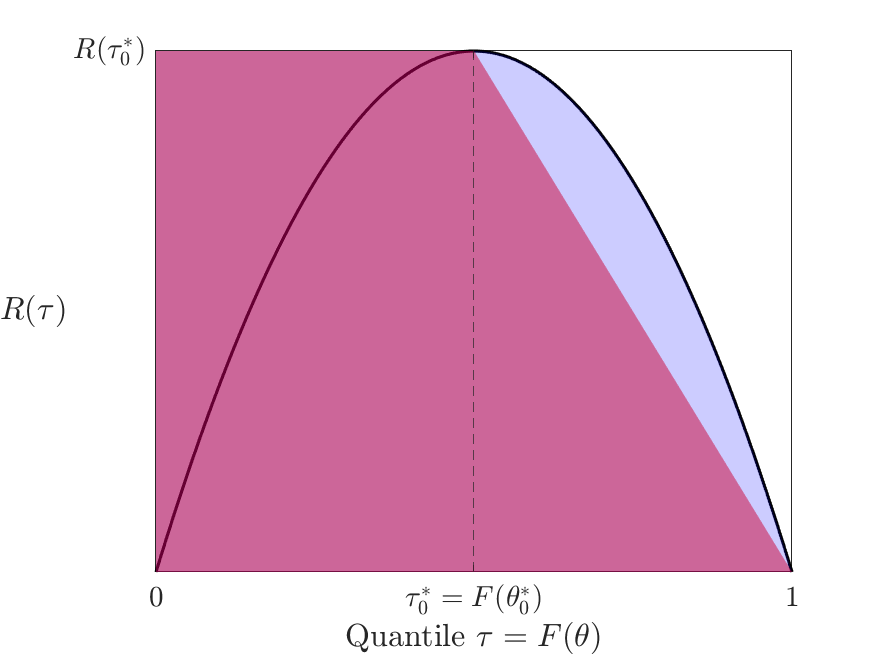}
    \end{subfigure}
    \caption{Revenue curve $R(\tau)$ for the uniform distribution}
    \label{fig:approx}
\end{figure}

The \emph{red} trapezoid in the right panel represents the revenue from selling a single good at price $p_0^* = \frac{1+F(\theta_0^*)}{2} \theta_0^*$, which induces $\bar s(\theta) = \frac{1+F(\theta_0^*)}{2}\cdot \mathbf{1} [\theta\geq\theta_0^*]$.
This is because the revenue under $\bar s$ is 
\[\int_{\theta_0^*}^{\bar\theta} J(\theta) \frac{1+F(\theta_0^*)}{2} \d F(\theta) 
= \frac{1+\tau^*_0}{2}R(\tau^*_0),\]
which equals the area of the red trapezoid in the right panel.
Since the seller can set the optimal price and $p_0^*$ is a feasible option, the revenue from selling a single tier is at least the red area.

Finally, the entire box, which  has area $R(\tau^*_0)$, represents the maximum revenue from selling ordinary goods in the auxiliary problem.
It is straightforward to verify that the blue area is strictly smaller than the entire box, and that the red area is at least half of the entire box.
Hence, the red area is strictly larger than half of the blue area.

The approximation is close for many common distributions, as shown in the following examples with $v(\theta)=0$.

\begin{example*}[Exponential distribution]
    If $F(\theta)=1-\exp(-\lambda\theta)$ where $\lambda>0$, selling a single good can obtain $91.9\%$ of the maximum revenue.
\end{example*}
\begin{example*}[Uniform distribution]
    If $F(\theta)=\theta$ on $[0,1]$, selling a single good can obtain $92.4\%$ of the maximum revenue. %
\end{example*}
\begin{example*}[Power distribution]
    Assume $F(\theta)=\theta^\beta$ on $[0,1]$.
    If $\beta\geq 1$, $J(\theta)$ is strictly increasing, so the optimal mechanism fully separates participants.

    If $\beta< 1$, because $J(\theta)$ is strictly single-crossing (from below) and increasing whenever $J(\theta) > 0$, the optimal mechanism still fully separates participants.
    In both cases, the approximation ratio is
    \begin{equation*}
        \frac{R_1}{R^*} =\frac{(1+\beta)\left(\frac{1}{1+2\beta}\right)^{\frac{1}{2\beta}}}{\beta+(1+\beta)^{-1-\frac{1}{\beta}}}   > 91.4\%.
    \end{equation*}
    The ratio approaches 1 as  $\beta\to \infty$ or $\beta\to0$ (i.e., $F$ is sufficiently convex or concave).
\end{example*}

\begin{example*}[Pareto distribution]
    If $F(\theta) = 1 - \theta^{-\beta}$ on $[1,\infty)$ where $\beta>1$ (so that the mean is finite), the approximation ratio is
    $$\frac{R_1}{R^*} = \left( 1 -\frac{1}{2\beta - 1} \right)^{\frac{\beta - 1}{\beta}}>79.29\%.$$
\end{example*}

\subsection{Welfare Comparative Statics}
\label{sec:welfare}

I now analyze how consumer welfare responds to regulations or other exogenous constraints on the number of levels and service coverage, given that the seller maximizes revenue subject to these constraints.

The directions of these welfare effects are \emph{a priori} ambiguous.
When more levels of positional goods are available, the seller may offer a finer status profile, which generates efficiency gains by assigning higher status to consumers who value it more.
However, as the seller charges higher prices for higher levels, she may also extract more surplus than she creates, potentially reducing consumer welfare.

Expanding coverage (by reducing the price of the lowest tier) also creates a tradeoff. 
When the seller offers a single tier, lowering its price has two effects.
On the extensive margin, it always benefits new participants who previously could not afford the good.
On the intensive margin, although existing consumers face lower prices, expanding coverage also reduces their status because they need to be pooled with new consumers.
For example, when more buyers gain access to a luxury good, the status 
conferred by the entry-level good may decline. When a service provider 
serves more customers, those at the basic paid tier may face longer wait times.

Moreover, these two effects are interdependent. When the seller offers more tiers, she may attract additional consumers, thereby expanding coverage. Conversely, when she expands coverage by lowering the price of the lowest tier, both consumers' choices and the seller's pricing among the higher tiers may vary, which in turn changes the status profile. 
To isolate the welfare effects of each policy dimension, I conduct comparative statics by varying one dimension at a time, under assumptions that fix the other dimension.

Using the envelope theorem, consumer surplus can be written as
\begin{align} \label{eqn:W}
   W =  \int_{\theta_0}^{\bar \theta} U(\theta) \d F(\theta) 
   =  \int_{\theta_0}^{\bar \theta} \left( \frac{1-F(\theta)}{f(\theta)}  (s(\theta)+ v'(\theta))  + U(\theta_0) \right)  \d F(\theta).
\end{align}
In the comparative statics, when the seller offers an optimal mechanism in response to policy changes, we have $U(\theta_0)=0$; alternatively,
when the lowest tier is free, we have $\theta_0=0$ and $U(0)=v(0)$.
Thus, the results are determined by the failure rate $\frac{f(\theta)}{1-F(\theta)}$.
Say $F$ satisfies IFR (DFR) if it has an increasing (decreasing) failure rate.

\begin{proposition}\label{prop:consumer}
    Assume $J(\theta)$ is increasing.
    The following comparative statics hold.
     \begin{enumerate}[label=(\roman*)]
        \item Assume $v(0)\geq\frac{v'(0)+1}{f(0)}$ or that the lowest tier is free.
        \footnote{If $v(0)\geq\frac{v'(0)+1}{f(0)}$, the optimal exclusion level is $\theta_0^*=0$.
        }
        Then, if $F$ satisfies IFR (DFR), consumer surplus decreases (increases) as the seller offers a finer status profile, and restricting the seller to a single tier increases (decreases) consumer surplus.

        \item 
        When the number of tiers is unrestricted, expanding coverage (by lowering the price of the lowest tier) increases consumer surplus.  

        When the seller is restricted to a single tier, expanding coverage increases consumer surplus if \(F\) satisfies IFR (and can decrease surplus otherwise).
    \end{enumerate}
\end{proposition}

The assumption in part (i) implies that the seller serves all consumers, either because it is optimal or it is required.
The result highlights the tension between revenue and consumer surplus under IFR\@: holding participation fixed, a finer status profile increases revenue but harms consumers.
Although a finer status profile creates efficiency gains, the seller extracts more surplus than she generates, leaving consumers worse off.
Therefore, when exclusion is fixed, restricting the seller to a single tier increases consumer surplus under IFR.
Nevertheless, consumer surplus need not decrease monotonically when more tiers are allowed, as raising the number of tiers does not necessarily lead the seller to a finer status profile.

The comparative statics are reversed if the distribution has a \emph{decreasing} failure rate.
Intuitively, a finer status profile shifts status from lower to higher types.
While higher types value status more, this shift does not automatically benefit them because the seller also extracts more surplus from them through higher prices.
DFR matters in two ways.
First, it makes types in the upper tail sufficiently dispersed, so improving the status of very high types creates information rents for them rather than being fully dissipated through higher prices.
Second, DFR implies that the weight $h(\theta)=\frac{1-F(\theta)}{f(\theta)}$ is increasing, so aggregate consumer surplus places greater weight on the gains of higher types than the losses of lower types.
Hence, a finer status profile raises aggregate consumer welfare. 

If consumer participation is \emph{not} fixed, the comparative statics become ambiguous because when the seller offers more levels, consumer participation may also increase, which affects consumer surplus.
For example, when $F(\theta)=\theta$ and $v(\theta)=0$, the seller excludes $\theta<1/\sqrt{3}\approx 0.58$ when constrained to a single tier, and excludes $\theta<0.5$ (while fully separating participants) when unconstrained.
Consequently, consumer surplus is higher in the latter case, as the gain from increased participation outweighs the loss from a finer status profile.
In general, the aggregate effect depends on the distribution of consumer types, and there is no simple condition for its sign.
This result highlights the difficulty of regulating the number of positional-good levels: unless the regulator also controls the exclusion level (for example, by mandating a free basic service), it is unclear whether such regulation benefits consumers.

When the number of levels is unconstrained, the seller fully separates participants if $F$ satisfies Myerson   regularity.
Thus, expanding coverage benefits new participants without lowering the status of existing ones.
Moreover, it benefits existing participants by requiring them to pay less for the same status.
Hence, consumers benefit unambiguously.

When the seller is constrained to offer a single tier, the welfare effect becomes more complicated.
If $F$ satisfies IFR, expanding coverage by lowering its price increases consumer surplus,
but for general distributions, expanding coverage may \emph{decrease} consumer surplus.
The reason is that while expanding coverage benefits the additional participants, it also reduces the status of current participants by pooling them with the newcomers.
When the type distribution has a thin tail, implied by IFR, the gain on the extensive margin outweighs the loss to current participants on the intensive margin.
In general, however, the loss to current participants may dominate, leading to a decrease in consumer surplus.

Finally, when the seller is restricted to $n>1$ tiers, the welfare effect of expanding coverage is ambiguous even under IFR. The reason is that when the seller expands coverage by lowering the price of the lowest tier, the prices and cutoffs of the higher tiers can also vary, thereby changing the status of consumers at those tiers. The boundary cases  $n=\infty$ and $n=1$ shut down this other dimension by fixing the status of higher-tier consumers (or making it vacuous), and are therefore tractable.

\section{Welfare Maximization}
In this section, I study the optimal mechanism that maximizes welfare.
Without positional externalities, the welfare-maximizing mechanism is to allocate the good to all consumers.
For positional goods, however, the value of the good diminishes as more consumers receive it, so the welfare-maximizing mechanism may involve offering multiple levels of goods and excluding some consumers.

First, I consider consumer surplus maximization when negative transfers (i.e., subsidies) are allowed and subject to budget balance.
Then, I study consumer surplus maximization under the constraint that transfers are nonnegative, which is relevant for applications where subsidies are unavailable or where the price represents effort.
Finally, I characterize the optimal mechanism that maximizes social welfare---a weighted sum of revenue and consumer surplus.

\subsection{Consumer Surplus Maximization}

I first consider consumer surplus maximization when negative transfers are allowed and subject to budget balance (i.e., $\int_{\theta_0}^{\bar \theta} p(\theta) \d F(\theta)=0$).
Then, substituting the budget balance constraint, consumer surplus becomes
\begin{align} \label{eqn:W2}
   W  =  \int_{\theta_0}^{\bar \theta} U(\theta) \d F(\theta) 
      = \int_{\theta_0}^{\bar \theta} (\theta s(\theta)  + v(\theta))   \d F(\theta).
\end{align}
It is straightforward that consumer surplus is maximized by full separation without exclusion and with cross-subsidization.
Intuitively, this is because assortative matching is efficient, and a fixed subsidy can redistribute efficiency gains among consumers while maintaining incentive compatibility.%
\begin{observation}\label{obs:1}
    When negative transfers are allowed and subject to budget balance, consumer surplus is maximized by $s(\theta)=F(\theta)$ and $p(\theta) = \int_0^\theta {t} \d F(t) - \E[(1-F(\theta)) \theta]$. 
\end{observation}

Under this mechanism, higher types pay more to attain higher status than lower types.
The payments will be used to cross-subsidize lower types, who accept lower status but are compensated by a subsidy. This mechanism can be implemented through an all-pay auction (without a reserve price) with a lump-sum subsidy of $\E[(1-F(\theta)) \theta]$, where buyers who pay more receive higher positions, and every buyer receives the same subsidy financed by the auction payments.
For sufficiently low types, the transfer $p(\theta)$ is negative, resulting in a net subsidy.

\subsection{Consumer Surplus Maximization without Subsidies}
\label{sec:CSmax_no_subsidy}

The result above relies on negative transfers, which may be infeasible in many applications.
In particular, for luxury or status goods, it is implausible to use subsidies.
Moreover, in applications to organizational hierarchies and education, the ``price'' represents effort and must be nonnegative.
Imposing a nonnegative price constraint $p(\theta)\geq 0$ significantly changes the optimal mechanism.

To maximize consumer surplus under the nonnegative price constraint, we need to pin down $U(0)$ in equation~\eqref{eqn:W}: 
\begin{align}  
   W =  \int_{\theta_0}^{\bar \theta} \left( \frac{1-F(\theta)}{f(\theta)}  (s(\theta)+ v'(\theta))  + U(\theta_0) \right)  \d F(\theta).
\end{align}
If the cutoff type $\theta_0>0$, then $U(\theta_0)=0$.
Otherwise, if $\theta_0=0$, because $U(0)=v(0)-p(0)$, consumer surplus  under the {nonnegative} price constraint ($p(\theta) \geq 0$) is maximized by $p(0)=0$ (so that $U(0) = v(0)$ and $p(\theta) = \int_{0}^\theta {t} \d s(t)\geq 0$). 
Since $U(\theta_0)$ is pinned down, the optimal mechanism is again determined by the failure rate.

\begin{proposition}\label{prop:welfaremax}
    Under the nonnegative price constraint, consumer surplus is maximized by
    \begin{enumerate}[label=(\roman*)]
    \item total pooling (i.e., $s(\theta)=1/2$) and $p(\theta)=0$ if $F$ satisfies IFR; %
    \footnote{
    The necessary and sufficient condition is 
    \[
    \int_0^\theta \left( \frac{1-F(t)}{f(t)} -\E[\theta] \right)\d F(t) \geq0 \quad \mbox{ for all $\theta\in[0,\bar \theta]$.}
    \]
    }
    \item full separation (i.e., $s(\theta)=F(\theta)$) and $p(\theta)=\int_0^\theta {t} \d F(t)$ if and only if $F$ satisfies DFR;%
    \item bottom pooling and top separation if the failure rate is single-peaked;
    \item bottom separation and top pooling if the failure rate is single-dipped.
    \end{enumerate}
    In (iii) and (iv), either pooling or separating region may be empty, and 
        $p(\theta)=\int_0^\theta t \d s(t)$.
\end{proposition}

\begin{remark}
    \label{complications}
    If the lowest type is changed to $\lbar\theta>0$, then pooling at the bottom can increase consumer surplus even when the failure rate is decreasing, as pooling increases $U(\lbar\theta) = \lbar\theta  s(\lbar\theta) + v(\lbar\theta)>0$ by raising the lowest status $s(\lbar\theta)$.
\end{remark}

Under the nonnegative price constraint, the consumer-optimal mechanism does not exclude consumers.
The result is in contrast to Proposition~\ref{prop:consumer}, where exclusion may increase consumer surplus under DFR\@, because Proposition~\ref{prop:welfaremax} focuses on the consumer-optimal mechanism.
Under IFR, because it is optimal to pool participants, exclusion raises the status of current participants, who would otherwise be pooled with more consumers.
Nevertheless, the loss to excluded consumers on the extensive margin exceeds the gain to participants on the intensive margin, so the net effect is detrimental.
Under DFR, exclusion has no benefit to current participants because the optimal mechanism does not pool them with additional participants.
\footnote{Although exclusion also changes prices, the price adjustments are
pinned down by incentive compatibility through revenue equivalence, so
the welfare comparison need only consider status and the exclusion level.}

Under IFR, the consumer-optimal regulation of positional goods without using subsidies is to serve all consumers with the same level of positional goods for free.
The results reverse under DFR, for the same reason as Proposition~\ref{prop:consumer}.  
For nonmonotone failure rates, the consumer-optimal mechanism can have pooling and separating regions. When the failure rate is single-peaked, such as the log-normal distribution, the optimal mechanism pools low types and separates high types.
When the failure rate is single-dipped, such as the power distribution $F(\theta)=\theta^\beta$ with $\beta\in(0,1)$, the optimal mechanism separates low types and pools high types.

The results also have implications for organizational hierarchies and education.
If the ability distribution has a thin tail, eliminating status competition entirely is optimal for agents' welfare.
\footnote{
Because the agent's payoff is scaled by $\theta$, the relevant distributional conditions are imposed on 
$\tilde f(\theta)\propto f(\theta)/\theta$ instead.
Thus, IFR of $\tilde F$ corresponds to a relatively thin upper tail of $F$ after weighting types by $1/\theta$, while DFR corresponds to a relatively thick upper tail under the same weighting.
}
In the education application, this implies complete educational disarmament---for example, randomly assign students to schools.
If instead there are a few superstars in the upper tail (e.g., under generalized Pareto distributions), meritocracy is optimal for agents' welfare.
\footnote{
    A generalized Pareto distribution with shape parameter $\xi\in(0,1)$ has a distribution function $F(\theta)= 1 - (1+\xi \theta)^{-1/\xi}$ on the support $[0, \infty)$,
    which has finite mean and satisfies DFR\@.
}
Intuitively, while meritocracy creates a rat race that forces all agents to exert effort, superstars can benefit from it because they have lower marginal costs of effort and are sufficiently spread out in the upper tail.
Although low-ability agents still suffer, the heavy tail also makes the gain to higher-ability agents dominate in the aggregate.
Hence, more (less) intense competition is optimal for agents' welfare when the ability distribution has a heavy (thin) tail.

The dependence on the failure rate is reminiscent of optimal bidding rings without transfers in \cite{McAfeeMcMillan1992}. Indeed, both results stem from assumptions that pin down $U(\theta_0)$.
An analog of their results also applies to this setting if agents can collude and transfers are impossible. In the application to organizational hierarchies or education, optimal collusion is collective disarmament under IFR, where all agents invest the minimal effort and share the same status, and status competition under DFR\@.

\subsection{Social Welfare Maximization}
In many settings, a social planner maximizes a weighted sum of the seller's revenue and consumer surplus.
Let $W_S = \E [\lambda p(\theta) + U(\theta)]$ (where $\lambda\geq0$) denote expected social welfare.
In applications to monopolist provision of positional goods, $p$ is monetary transfer, and $\lambda$ is the welfare weight on the seller's revenue relative to consumer surplus. 
When $\lambda\geq1$, the social planner puts higher weight on the seller's revenue than consumer surplus.

In applications to education or status within organizations, $p\geq0$ is nonnegative effort, and $\lambda$ captures the degree to which effort is productive from the social perspective.
When $\lambda=0$, the agents' effort is purely signaling.
As $\lambda$ increases, effort becomes more productive, and the social planner has greater incentives to induce agents to exert effort.

If negative transfers are allowed and subject to budget balance ($\E[p(\theta)]=0$), then $W_S = \E[U(\theta)]$, so full separation without exclusion, as shown in Observation~\ref{obs:1}, also maximizes social welfare.
Therefore, the more interesting case to consider is when budget balance is not required, or when negative transfers are not allowed. In the former case, I assume the welfare weight $\lambda\geq1$, since otherwise a transfer from the seller to consumers would always increase social welfare without bound.

Define the weighted virtual value as $J_\lambda(\theta) = 
\lambda J(\theta) + h(\theta)= \lambda \theta - (\lambda-1)\frac{1-F(\theta)}{f(\theta)}$.
The expected social welfare is
\begin{equation*}
        W_S = \int_{\theta_0}^{\bar \theta}J_\lambda(\theta) s(\theta) \d F(\theta) +
        \int_{\theta_0}^{\bar \theta} \left( \lambda v(\theta)-(\lambda-1) v'(\theta) \frac{1-F(\theta)}{f(\theta)} - (\lambda-1)U(\theta_0) \right) \d F(\theta).  
\end{equation*}
In particular, if $\lambda=1$, transfers are welfare-neutral, and we have $J_\lambda(\theta) = \theta$. %
The social welfare is maximized by full separation without exclusion because a finer status profile strictly increases social welfare by matching the higher types with higher status.

If $\lambda>1$, the social planner puts higher weight on the revenue, so it is optimal to set $U(\theta_0)=0$, and as in revenue maximization, the payment identity implies that the price is nonnegative.
If $\lambda<1$ and negative transfers are not allowed, analogous to the consumer surplus maximization without subsidies (Section~\ref{sec:CSmax_no_subsidy}), we can also pin down $U(0) = v(0)$ if $\theta_0=0$ and $U(\theta_0)=0$ if $\theta_0>0$.
Therefore, we have the following result.

\begin{proposition}\label{prop:social}
    Assume that $\lambda\geq1$ or that negative transfers are not allowed.
    If $J_\lambda(\theta)$ is increasing, the social welfare-maximizing mechanism fully separates the participants and excludes types below a certain threshold.  
\end{proposition}

In general, the social welfare-maximizing mechanism depends on the welfare weight $\lambda\geq0$ on the seller relative to consumers, or the extent to which effort is productive relative to signaling.

\section{Extensions}
\label{sec:extensions}

\subsection{Alternative Specifications of Status}
\subsubsection{Convex or Concave Status}
When positional externalities arise from interpersonal comparison, a concave or convex specification of status may arise. Assume $\phi\colon [0,1] \to [0,1]$ is strictly increasing, continuously differentiable, and satisfies $\phi(0)=0$, and define the status function as
\begin{equation} \label{eqn:status_phi}
   \tilde S({x}, G)= \phi(S({x}, G)) 
    =  \phi \left( \frac{G^{-}({x}) + G({x})}{2} \right).
\end{equation}
Thus, the (interim) status profile is given by
\(
\tilde s(\theta)
=
\E_\omega\left[
\phi\left(
S(\chi(\theta,\omega),G_\chi(\cdot\mid\omega))
\right)
\right].
\)
Let
$$\mathcal F(\theta_0) = \{\tilde s \colon \Theta \to [0,1] \mid  \tilde s  = \phi \circ s,\;    s\in \ext \MPS(F\cdot\mathbf{1}_{[\theta_0,\bar\theta]}),\; s(\theta)=0 \mbox{ for all }\theta\in [0,\theta_0) \}$$ 
denote the set of status profiles obtained by transforming the extreme points of the original feasibility set for a given $\theta_0$.
Let $L^1$ denote the set of $F$-integrable functions $g\colon\Theta\to\mathbb R$ endowed with the norm $
\|g\|_{L^1}
=
\int_\Theta |g(\theta)|\d F(\theta)$.
Let $\overline{\conv} \mathcal F(\theta_0)$ denote the closed convex hull of $\mathcal F(\theta_0)$ under the $L^1$-norm topology.
\footnote{
Formally, define \(
\conv \mathcal F(\theta_0)
=
\left\{
\sum_{i=1}^n \lambda_i \tilde s_i
\;\middle|\;
n\in\mathbb N,\ 
\tilde s_i\in\mathcal F(\theta_0),\ 
\lambda_i\ge 0,\ 
\sum_{i=1}^n \lambda_i=1
\right\}
\)
and define $\overline{\conv} \mathcal F(\theta_0)$ as its closure under the $L^1$-norm topology.
}
Under this specification, an increasing status profile $\tilde{s}$ is feasible if and only if $\tilde s \in \bigcup_{\theta_0\in[0,\bar\theta]} \overline{\conv} \mathcal F(\theta_0)$.

The following proposition shows that the main results are robust to nonlinear status, although stronger regularity conditions may be required.
The proof relies on two observations. First, because the objective is a continuous linear functional of $\tilde s$ (with respect to the $L^1$-norm), the supremum over $\mathcal F(\theta_0)$ equals the supremum over its closed convex hull.

Second, instead of optimizing directly over $\mathcal F(\theta_0)$, 
the proof considers the relaxed problem over a larger set $\{\tilde s  \mid  \tilde s  = \phi \circ s,\;    s\in \MPS(F\cdot\mathbf{1}_{[\theta_0,\bar\theta]})\} \supseteq \mathcal F(\theta_0)$, obtained by dropping the extreme-point restriction.
Then, the proof verifies that the solution to the relaxed problem belongs to the feasible set $\mathcal F(\theta_0)$---that is, the optimal allocation is deterministic.
Hence, the relaxation is tight.
\footnote{
In general, $\E_\omega[\phi(\cdot)]
\neq
\phi\left(\E_\omega[\cdot]\right)$.
However, because a deterministic allocation is optimal, the randomization variable $\omega$ is degenerate, and the two expressions coincide. Therefore, the same conditions and conclusions apply to the alternative specification in which status is linear (i.e., $S(x,G)= \frac12 (G^-(x)+G(x))$), but the agent's payoff is nonlinear in status: 
$u(p,s,\theta) = \theta\phi(s) - p + v(\theta)$.
}

\begin{proposition}\label{prop:convex_concave}
    Assume $\phi$ is continuously differentiable and strictly increasing, and that $v(\theta)=\alpha\theta$ for some $\alpha\geq 0$.
    \begin{enumerate} 
        \item If, in addition, $\phi$ is convex, we have the following. 
    \begin{itemize}
        \item The revenue-maximizing mechanism fully separates participants and excludes $\theta<\theta_0^J$ if $J(\theta)$ is increasing.
        \item  Holding exclusion fixed, as the status profile becomes finer, consumer surplus increases if $F$ satisfies DFR.
        \item Under the nonnegative price constraint, consumer surplus is maximized by full separation under DFR.
        Without the constraint, Observation~\ref{obs:1} continues to hold.
        \item Assume that $\lambda\geq1$ or that negative transfers are not allowed. The social  welfare-maximizing mechanism fully separates participants if $J_\lambda(\theta)$ is increasing.
   \end{itemize}
    
    \item If, in addition, $\phi$ is concave, we have the following.
    \begin{enumerate}
        \item The revenue-maximizing mechanism fully separates participants if $J(\theta) \phi'(F(\theta))$ is increasing.
        \item  Holding exclusion fixed, as the status profile becomes finer, consumer surplus decreases if $F$ satisfies IFR.
        \item  Under the nonnegative price constraint, consumer surplus is maximized by total pooling under IFR, and maximized by full separation if $\frac{1-F(\theta)}{f(\theta)}\phi'(F(\theta))$ is increasing.

        \item  Assume that $\lambda\geq1$ or that negative transfers are not allowed. The social  welfare-maximizing mechanism fully separates participants if $J_\lambda(\theta)\phi'(F(\theta))$ is increasing.
    \end{enumerate}
    \end{enumerate}
\end{proposition}

\begin{remark}
    The conditions also apply to the specification where the status remains linear (i.e., $S(x,G)= \frac12 (G^-(x)+G(x))$), but the agent's payoff is strictly increasing and convex or concave in status---i.e., $u(p,s,\theta) = \theta\phi(s) - p + v(\theta)$.
\end{remark}

When $\phi$ is convex, the results for full separation follow from Fan--Lorentz Theorem, as $\Phi(z,s) = z\cdot\phi(s)$ is convex in $s$ and supermodular in $(s,\theta)$
for $z(\theta)\in\{J(\theta), J_\lambda(\theta), \frac{1-F(\theta)}{f(\theta)}\}$ if $z(\theta)$ is increasing and nonnegative.
In particular, the assumption $v(\theta)=\alpha\theta$ implies that the optimal exclusion
thresholds for revenue maximization and social-welfare maximization are
$\theta_0^J = \inf\{\theta\in\Theta:J(\theta)>0\}$ and
$\theta_0^\lambda = \inf\{\theta\in\Theta:J_\lambda(\theta)>0\}$,
respectively. Hence, $J(\theta)$ and $J_\lambda(\theta)$ are nonnegative
on their respective participation intervals.

When $\phi$ is concave, the results under IFR also follow from Fan--Lorentz because $\frac{1-F(\theta)}{f(\theta)}\phi(s)$ is concave in $s$ and submodular in $(s,\theta)$.
In the remaining cases, the conditions involve an increasing $\phi'(F(\theta))$.
Because $J(\theta)$ being increasing does not necessarily imply that $J(\theta)\phi'(F(\theta))$ is increasing, the revenue-maximizing mechanism may involve pooling even under Myerson regularity.

\subsubsection{Varying Positional Concerns about Same-tier Consumption}
\label{sec:extensions:gamma}

Now consider a more general specification à la \citet{HopkinsKornienko2004}:
\[ S({x}, G(\cdot)) = \gamma G({x}) + (1-\gamma) G^{-}({x}), \]
where $\gamma\in [0,1]$ measures the intensity of concern about the mass of consumers at the same level.
The benchmark case $\gamma=1/2$ ensures that pooling preserves the total status.
Hence, for a given level of exclusion, the participants' expected total status is invariant to the number of status levels.
If $\gamma<1/2$ ($\gamma>1/2$), consumers discount same-level consumers more (less) heavily, so pooling decreases (increases) expected aggregate status and the feasibility condition in Theorem~\ref{thm:feasibility} needs to be modified accordingly.

Two extreme cases, $\gamma=0$ and $\gamma=1$, admit simple characterizations of feasibility using monotone function intervals \citep{YangZentefis2024}. When $\gamma=0$, consumers only derive utility from the mass of consumers at \emph{strictly} lower levels and discount same-level consumers entirely. 
\footnote{This corresponds to the linear case in \cite{ImmorlicaStoddardSyrgkanis2015}.}
In this case, an increasing status profile $s$ is feasible if and only if $s\in\mathcal I(0,F)$, that is,
\[
0\leq s(\theta)\leq F(\theta) \quad \text{for all } \theta\in[0,\bar\theta].
\]
Intuitively, pooling strictly reduces the status of pooled types without affecting the status of other types.
The following result characterizes the optimal mechanisms, which fully separates buyers except for types that are excluded in the benchmark case $\gamma=1/2$.
\footnote{
In the absence of the intrinsic value $v(\theta)$, exclusion is equivalent to pooling at the bottom tier because both lead to $s(\theta)=0$.
}

\begin{observation} \label{obs:gamma=0}
    Assume $\gamma=0$. The optimal mechanisms are as follows:
    \begin{enumerate}[label=(\roman*)]
        \item Consumer surplus, whether negative transfers (subject to budget balance) are allowed or not, is maximized by full separation (i.e., $s(\theta)=F(\theta)$).%
        \item If $J(\theta)$ is strictly single-crossing (from below) and $v(\theta)=\alpha\theta$, the revenue-maximizing mechanism excludes $\theta<\theta_0^J$ and fully separates participants.
    \end{enumerate}
\end{observation}

At the other extreme, where $\gamma=1$ (see \cite{Frank1985}), consumers do not distinguish same-level consumers from lower-level ones and derive utility from them equally.
In this case,  an increasing status profile $s$ is feasible if and only if
there exists $\theta_0\in[0,\bar\theta]$ such that 
$$F(\theta)\le s(\theta)\le 1\quad \text{for all } \theta\in[\theta_0,\bar\theta],$$ 
and $s(\theta)=0$ for all $\theta\in[0,\theta_0)$.
Since $\gamma=1$ implies $S({x}, G(\cdot)) = G({x})$, pooling all participants assigns them the highest status $s=1$ without diluting it, as in selling ordinary goods without positional externalities.
Therefore, a single-tier posted-price mechanism that pools all participants is optimal for consumer surplus.
\begin{observation} \label{obs:gamma=1}
    Assume $\gamma=1$. The optimal mechanisms pool all participants, regardless of the distribution:
      \begin{enumerate}[label=(\roman*)]
        \item
        Consumer surplus, whether negative transfers (subject to budget balance) are allowed or not, is maximized by total pooling without exclusion (i.e., $s(\theta)=1$).
\item 
   If $v(\theta)=\alpha\theta$, the revenue-maximizing mechanism is a single-tier posted-price mechanism, i.e., $s(\theta)=\mathbf{1}_{\theta\geq\theta_0}$ for some $\theta_0\in[0,\bar\theta)$.
    \end{enumerate}
\end{observation}

In general, when $\gamma\in(0,1)\setminus\{1/2\}$, the feasibility condition is more complex, but
the baseline case $\gamma=1/2$ provides a useful benchmark.
For consumer surplus maximization, we have the following: if full separation is optimal under the $\gamma=1/2$ benchmark, then it remains optimal for $\gamma\in[0,1/2]$; if total
pooling is optimal under the benchmark, then it remains optimal for
$\gamma\in[1/2,1]$. Intuitively, $\gamma>1/2$ raises the status assigned
to pooled types, making pooling more attractive, while $\gamma<1/2$
lowers the status assigned to pooled types, making pooling less
attractive.
For revenue maximization, the same conclusion requires $J(\theta)\geq0$ on the participating region $[\theta_0,\bar\theta]$.
A sufficient condition is that $J(\theta)$ strictly single-crosses zero from
below and $v(\theta)=\alpha\theta$, which implies
$
\theta_0\geq \theta_0^J
$ and therefore $J(\theta)\geq0$ on $[\theta_0,\bar\theta]$.
Otherwise, when some participating types have negative virtual value, lowering their
status through pooling can increase revenue.

\subsubsection{Signaling}
\label{sec:signaling}

Consider an alternative specification where status arises from signaling concerns.
For example, consumers purchase conspicuous goods ${x}$ to obtain social status valued at $\E[\theta| {x}]$ \citep{Rayo2013}. Under this specification, the status induced by a positional good allocation $\chi\colon \Theta \times \Omega \to \tilde{X}$ is given by
\[S({x}, G_\chi) = \E[\theta \mid \chi(\theta,\omega) = {x}].\]
An important implication of this specification is that,
even if a consumer does not purchase the good (i.e., ${\sigma(\theta)} = 0$), he still obtains a baseline status $\lbar s =\E[\theta \mid {\chi(\theta)} = 0]$ rather than zero as in the main specification. 

This feature has two implications. First, the feasibility condition in Theorem~\ref{thm:feasibility} becomes analogous to the case without exclusion (see Remark~\ref{MPS(F)}): 
an incentive-compatible (and hence increasing) $s(\theta)$ is feasible if and only if $s\in \MPS(\theta)$ in quantile space, that is,
\begin{equation}
    \int_\theta^{\bar\theta} s(t) \d F(t) \leq \int_\theta^{\bar\theta} t \d F(t) \quad \mbox{ for all $\theta\in[0,\bar\theta]$}
\end{equation}
with equality at $\theta=0$.
In terms of status profiles, being excluded is equivalent to being pooled in a single bottom tier, although the buyer obtains an intrinsic value $v(\theta)$ in the latter case.

Moreover, the buyer's outside option becomes $\theta\lbar s$ instead of zero.
When the seller excludes types $\theta<\theta_0$, the baseline status is $\lbar s = \E[\theta \mid \theta<\theta_0]$.
If everyone participates, assume that the off-path belief following nonparticipation assigns probability one to the lowest type, so that $\lbar s=0$.
Thus, when the intrinsic value $v(\theta)=0$, instead of excluding types at the bottom, the seller can achieve the same revenue by pooling them into a single bottom tier, which replicates the same status level $\lbar s$ without affecting the incentives of higher types.
Hence, it is without loss of optimality to assume no exclusion; in other words, exclusion cannot increase revenue through the status channel, in contrast to the main specification, because nonparticipants still obtain the baseline status $\lbar s$.
\footnote{
Of course, if the intrinsic value $v(\theta)>0$, then exclusion may increase revenue through the standard channel.
}

Consequently, if $v(\theta)=0$, the main results still hold except that there is no exclusion, analogous to the case where exclusion is assumed to be impossible (see Appendix~\ref{app:noexclusion}): the revenue-maximizing mechanism fully separates all types if and only if $J(\theta)$ is increasing, while consumer surplus decreases (increases) as the seller offers more levels if $F$ has an increasing (decreasing) failure rate.

\subsection{Intrinsic Quality}
\label{intrinsic}
In the baseline model, to abstract from quality differentiation, I assume that positional goods have the same intrinsic quality across tiers and thus deliver the same intrinsic value $v(\theta)$.
In many applications, however, a higher-tier positional good, such as a larger car or a higher boarding class, also comes with a higher intrinsic quality.

In this subsection, I assume that the tier ${x}$ is not only ordinal. Instead, a good ${x}\in \tilde X = \reals_{+}$ has intrinsic quality equal to $x$, which delivers intrinsic value $v({x},\theta) = \theta {x}$ and incurs a production cost $c({x})$, à la \citet{MussaRosen1978}.
Assume that $c({x})$ is continuously differentiable, strictly increasing, strictly convex, and satisfies $c(0) = c'(0) = 0$ and $\lim_{x \to \infty} c'(x) = \infty$.
Because $v(0,\theta) = 0$ and $c(0) = 0$, it is still without loss of generality to denote the outside option of not buying by $x=0$.

As before, consider a direct mechanism $(\chi(\theta,\omega), p(\theta))$ consisting of an allocation rule $\chi\colon \Theta \times \Omega \to \reals_+$ and a payment schedule $p\colon \Theta \to \reals$.
The participation decision is incorporated into the allocation $\chi$, as $\chi(\theta,\omega) = 0$ is equivalent to opting out.
Let ${x}(\theta) = \E_\omega[\chi(\theta,\omega)]$ denote the interim allocation.
Consider the direct mechanism $(s(\theta), {x}(\theta), p(\theta))$, where $s(\theta)$ is the status profile induced by $\chi$.
Thus, the buyer's utility is $U(\theta) = \theta(s(\theta)+{x}(\theta))-p(\theta)$.
By convention, the opt-out decision is $\chi(\theta) = 0$ (which induces $s(\theta) = 0$) and $p(\theta) = 0$.
The following lemma on incentive compatibility replaces Lemma~\ref{lemma:IC}.

\begin{lemma}
    A direct mechanism $(s(\theta), {x}(\theta), p(\theta))$ is incentive-compatible if and only if 
    \begin{itemize}
     \item $s(\theta)+{x}(\theta)$ is increasing;
     \item $U(\theta) = U(0) + \int_{0}^\theta \left( s(t)+ {x}(t)\right) \d {t}$ for all $\theta \in[0,\bar \theta]$.
    \end{itemize}
\end{lemma}
In addition to incentive compatibility, the mechanism must satisfy the compatibility constraint that $s(\theta)$ and ${x}(\theta)$ must be generated by the same allocation $\chi$.
Say $s$ and ${x}$ are \emph{compatible} if $s(\theta) = \E_\omega[S(\chi(\theta,\omega),G_\chi)]$ and ${x}(\theta) = \E_\omega[\chi(\theta,\omega)]$ for some allocation $\chi\colon \Theta \times \Omega \to \reals_+$.
Using the same arguments as before and setting $U(0)=0$, the revenue maximization problem is 
\begin{equation}
    \max_{\chi\colon \Theta \times \Omega \to \reals_+} \int_{0}^{\bar \theta}   \left( J(\theta) (s(\theta)+{x}(\theta)) - \E_\omega[c(\chi(\theta,\omega))] \right) \d F(\theta)
\end{equation}
subject to monotonicity and compatibility constraints.
Under the regularity condition that $J(\theta)$ is increasing, there is no loss of optimality in restricting attention to deterministic increasing allocations
$\chi\colon \Theta \to \reals_+$
(see Lemma~\ref{lemma:compatibility}).

In general, because of the compatibility constraint, the optimization over status $s$ and quality $x$ cannot be separated. However, in this case, the separate optimizers over $s$ and $x$ happen to be compatible, as they can be generated by the same deterministic allocation $\chi$.

\begin{proposition} \label{prop:intrinsic}
    Assume $J(\theta)$ is increasing.
    The profit-maximizing mechanism excludes $\theta\leq\theta_0^J$  (i.e., $
    \chi^*(\theta) = 0$ for all $\theta\leq\theta_0^J$) and assigns $\chi^*(\theta) = c'^{-1}( J(\theta))$ for all $\theta>\theta_0^J$, which induces $x^*(\theta)=\chi^*(\theta)$ and 
    \[
    s^*(\theta) = \begin{cases}
        0, &\text{if } \theta\leq\theta_0^J\\
        \frac{F(\lbar\theta_i)+F(\bar\theta_i)}{2}, &\text{if } \theta\in [\lbar\theta_i,\bar\theta_i]\\
        F(\theta), &\text{otherwise}\\
    \end{cases}
    \]
    where $[\lbar\theta_i,\bar\theta_i]$ are the intervals where $J(\theta)$ is constant and positive.
\end{proposition}

The optimal status mechanism excludes types below $\theta_0^J$, for which the virtual valuation $J(\theta)$ is negative.
For types above $\theta_0^J$, if $J(\theta)$ is strictly increasing, the optimal mechanism separates all types and induces $s^*(\theta) = F(\theta)$; if $J(\theta)$ is constant, the optimal mechanism pools types with the same virtual valuation and induces the same status for them. 
The proposition implies that the profit-maximizing mechanism is robust to intrinsic quality, or equivalently, that the quality differentiation result in \citet{MussaRosen1978} is robust to positional externalities.

By attaching status value, the seller can extract $\int_{\theta_0^J}^{\bar \theta} J(\theta) s^*(\theta) \d F(\theta)>0$ additional revenue.
For the uniform distribution on $[0,1]$ and quadratic cost $c({x})={x}^2/2$, profit with status value is 0.29 compared to 0.08 with pure intrinsic value.
This may explain why luxury companies invest heavily in cultivating status value.

\subsection{Screening with Ordeals: Excessive Waiting Time}
\label{negative}
Priority service is an example of positional goods that arise from capacity constraints.
As shown in Example~\ref{ex:queuing}, the baseline model directly applies after defining the status by $s=1-t$, where $t$ is the waiting time.
In this subsection, I consider an extension where the service provider can delay service and prolong waiting time to $t(\theta) = 1-s(\theta) > 1$ based on the buyer's reported type $\theta$.

In this case, the range of status profiles becomes $(-\infty,1]$.
Incentive compatibility still requires the status profile $s(\theta)$ to be increasing, and the set of increasing and feasible status profiles is the lower set of the original feasible set, that is, 
$\{ s\colon \Theta\to (-\infty,1] \text{ increasing} \mid s \leq \hat{s} \text{ for some } \hat{s}\in \MPS_0(F)\}$.

Can the seller benefit from delaying service?
Excessive waiting time ($t>1$) can be viewed as an ordeal to screen low types.
Intuitively, instead of excluding low types, the seller can serve them with excessive waiting time, thereby extracting revenue from them while preventing higher types from mimicking them.
The following proposition confirms this intuition.

\begin{proposition} \label{prop:negstatus}
    Assume $J(\theta)$ is increasing and that delaying service is feasible.
    Then, the revenue-maximizing mechanism is given by 
    \[
            s^*(\theta) = \begin{cases}
                F(\theta), & \text{if } \theta\geq \theta_0^J\\
                -v'(\theta), & \text{if } \theta< \theta_0^J
            \end{cases} \mbox{ and } 
    p^*(\theta) = \begin{cases}
        \theta_0^J F(\theta_0^J) + \int_{\theta_0^J}^{\theta} z \d F(z) + v(\theta_0^J), & \text{if } \theta\geq \theta_0^J\\
        v(\theta)-\theta v'(\theta), & \text{if } \theta< \theta_0^J
    \end{cases}
    \]

   Allowing the seller to delay service increases revenue and decreases consumer surplus; both effects are strict if $f(0)<\infty$ and $v(\theta)/\theta > v'(\theta)$ for all $\theta>0$.
\end{proposition}

The optimal mechanism serves types with positive virtual valuations in descending order (i.e., $t^*(\theta)=1-F(\theta)$), with higher types served first, and assigns a service time of $t^*(\theta)=1+v'(\theta)\geq 1$ to types with negative virtual valuations.
In particular, $t^*(\theta)=1+v'(\theta)$ is decreasing and thus incentive-compatible because $v''(\theta)\leq 0$.

Delaying service increases revenue in two ways. First, the seller collects positive payments from types $\theta\in(0,\theta_0^*)$, who are excluded in the benchmark.
Second, for types $\theta\in (\theta_0^*, \theta_0^J)$ who are already served in the benchmark, because their virtual valuation $J(\theta)$ is negative, the seller can now extract strictly higher revenue from them by serving them with $t^*(\theta)=1+v'(\theta)$.
When the average intrinsic value is strictly higher than the marginal value (i.e., $v(\theta)/\theta > v'(\theta)$), the revenue gain from $\theta\in(0,\theta_0^*)$ is strict.
If $\theta_0^*=0$, the condition $f(0)<\infty$ implies $(\theta_0^*, \theta_0^J)$ is nonempty, so the revenue gain from types in this interval is also strict.
Consumers are indifferent between being excluded and served with $t^*(\theta)=1+v'(\theta)$, earning zero rents either way, so the first channel does not affect them.
However, the second channel increases the mass of consumers who receive zero rents, thereby reducing consumer surplus.

The following example illustrates how the seller can implement the optimal mechanism by pausing service.

\begin{example}[Delaying Service]
    Consider the payoff $u(p,s,\theta) = \tilde v (\theta)-\theta t - p$, where $t=1-s$ is the waiting time.
    Assume $\theta\sim \text{Unif }[0,1]$ and $\tilde v(\theta)=v_0+(1+\alpha)\theta$ (so that $v(\theta)= v_0+\alpha\theta$), where $v_0,\alpha\geq0$.
    Then, the revenue-maximizing mechanism is
    \[
        t^*(\theta) = \begin{cases}
            1-\theta, & \text{if } \theta\geq 1/2\\
            1+  \alpha, & \text{if } \theta<1/2
        \end{cases}
        \mbox{ and }
        p^*(\theta) = \begin{cases}
            \theta^2/2  + \alpha/2 + 1/8 + v_0, & \text{if } \theta\geq 1/2\\
            v_0, & \text{if } \theta<1/2.
        \end{cases}
    \]
    \begin{figure}[htbp]
        \centering
        \includegraphics[width=0.4\textwidth]{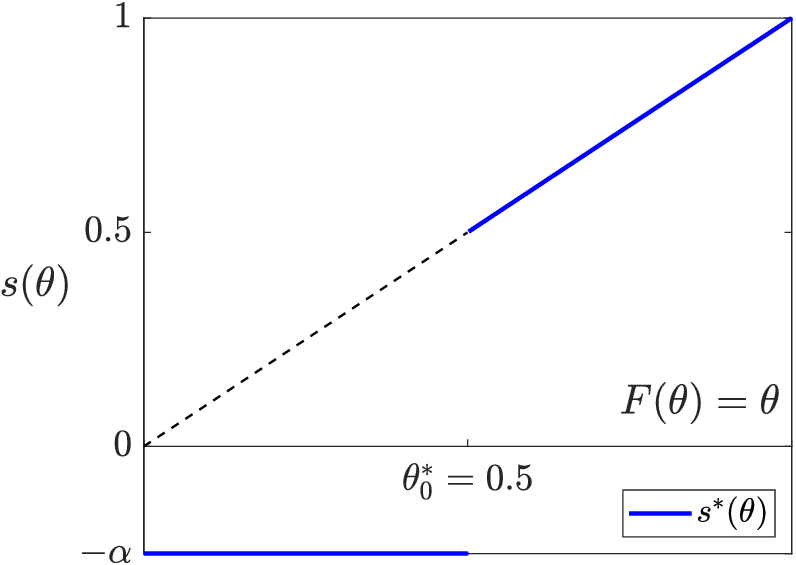}
        \hspace{25pt} \caption{Delaying Service}
    \end{figure}
    
    Instead of excluding low types, the seller assigns $t^*(\theta) = 1+\alpha$ to types with negative virtual valuations and charges $p^*(\theta) = v_0$, which can be implemented by pausing service for $\alpha+1/4$ units of time before serving types $\theta<1/2$ in random order (by pooling them into the same priority level).
\end{example}

The insight that excessive waiting time serves as a substitute for exclusion also applies to the case where exclusion is impossible, as the lowest-level good or service is free (see Appendix~\ref{app:noexclusion}).
In this case, excessive waiting time diminishes the value of the free service and effectively plays the role of exclusion, thereby increasing the seller's revenue.
In particular, by offering a priority service and a free regular service, the seller extracts more revenue as she increases the waiting time for the free regular service.

\subsection{When Higher Types Obtain Lower Utility}
\label{suffering}

Now I consider the case where $v(\theta)\geq0$ and $v'(\theta) \leq -1$, so that higher types obtain lower total utility from purchasing the good after accounting for status and intrinsic value.
This arises naturally in Example~\ref{ex:queuing}: if the value of the good is homogeneous across agents, higher types derive lower net utility after accounting for waiting costs.
\footnote{
In Example~\ref{ex:queuing}, if the good has the same value $v_0\geq\bar\theta$ to everyone, then the payoff $u(p,s,\theta) = v_0 - \theta t - p = (v_0 - \theta) + \theta s - p$ (where $s=1-t$) satisfies $v'(\theta) = -1$ and $v(\theta)\geq0$.
}
Therefore, the upward incentive constraint is the binding one: agents have incentives to \emph{overreport} in order to obtain shorter waiting times.
Intuitively, an agent’s type can be interpreted as his marginal cost of waiting, making the problem analogous to a procurement auction.

Formally, since \(s(\theta)\le 1\), when \(v'(\theta)\le -1\), the
envelope condition implies
\[
U'(\theta)=s(\theta)+v'(\theta) \le 0
\]
for all participating types $\theta$ (for which $\sigma(\theta)=1$).
Thus, in contrast to the main specification, higher types receive \emph{lower} information rents, so agents must be prevented from overreporting, and the local IC constraint binds from below rather than from above.
Moreover, it is now \emph{high} types that are potentially excluded: there exists a cutoff $\theta_0\in[0,\bar\theta]$ such that a buyer participates (i.e., $\sigma(\theta)=1$) for all $\theta<\theta_0$ and does not participate (i.e., $\sigma(\theta)=0$) for all $\theta>\theta_0$.

When participating buyers are fully separated, their status profile is $s(\theta) =  F(\theta) + (1-F(\theta_0))$.
In general, the feasibility condition in terms of majorization continues to hold, up to a shift, for all participants: an incentive-compatible $s$ is feasible if and only if there exists $\theta_0\in[0,\bar\theta]$ such that $$\int_{0}^{\theta} [s(t) - (1-F(\theta_0))] \d F(t) \geq \int_{0}^{\theta} F(t) \d F(t),\quad\mbox{ for all }\theta\in[0,\theta_0]$$
with equality at $\theta=\theta_0$, and $s(\theta) = 0$ for all $\theta\in(\theta_0, \bar\theta]$.

The following proposition summarizes the results for this setting.
\begin{proposition} \label{prop:suffering}
    The following results hold if $v(\theta)\geq0$ and $v'(\theta)\leq -1$:
    \begin{enumerate}[label=(\roman*)]
        \item %
        Assume $L(\theta) = \theta + \frac{F(\theta)}{f(\theta)}$ is increasing. Then,
        \begin{enumerate}
            \item The revenue is increasing in the number of positional good levels.
            \item 
            The revenue-maximizing mechanism excludes $\theta>\theta_0^*$ and fully separates $\theta\leq\theta_0^*$, where the optimal cutoff is $\theta^*_0 \in \argmax_{\theta_0}  \int_{0}^{\theta_0}   L(\theta) F(\theta) \d F(\theta) + [v(\theta_0) + \theta_0 (1-F(\theta_0))] F(\theta_0)$.
            \item If $v(\bar\theta)\geq-\frac{v'(\bar\theta)+1}{f(\bar\theta)}$ and $v$ is concave, then zero exclusion is optimal (i.e., $\theta_0^*=\bar\theta$).
        \end{enumerate}

        \item %
        The seller can obtain at least half the maximum revenue by selling a single tier.

        \item %
        Under the nonnegative price constraint, consumer surplus is maximized by total pooling (full separation) without exclusion under IFR (DFR).

        \item Delaying service cannot increase the seller's revenue.
    \end{enumerate}
\end{proposition}

Parts (i)--(iii) are qualitatively similar to those in the main
specification, except that
the revenue-maximizing mechanism excludes high types rather than
low types.

For part (iv), because $L(\theta)\geq0$, any downward distortion of status caused by delaying service reduces revenue. Intuitively, delaying service makes overreporting even more attractive and therefore increases the information rents required for incentive compatibility. Hence, unlike in the benchmark, excessive waiting cannot increase revenue.

\section{Conclusion}
Many goods are positional because consumers care about their relative consumption, either for psychological (such as luxury goods) or physical reasons (such as priority services).
Thus, allocating positional goods generates externalities on others: improving one consumer's position necessarily worsens another's, and the status value of the good diminishes as more consumers buy it.
Because of these externalities, the optimal mechanisms in terms of revenue and welfare are potentially different from those for ordinary goods.

In this paper, I study the allocation of positional goods using a mechanism design approach.
I characterize the optimal mechanism in terms of revenue, consumer surplus, and social welfare.
I study the welfare effects of restricting the seller to a single tier.
Under this restriction, I show that the seller can still guarantee at least half of the maximum revenue and that expanding coverage benefits consumers under IFR but may harm them otherwise.

The results also have implications for educational disarmament.
To maximize aggregate effort, meritocratic competition with potential exclusion is optimal.
For student welfare, however, meritocracy is harmful under common thin-tailed ability distributions, suggesting the need for disarmament policies, such as lottery-based school assignment or coarser performance rankings.
By contrast, when the distribution is heavy-tailed, meritocracy raises student welfare in the aggregate, as high-ability students benefit substantially from separation.
From a social welfare perspective, when effort is sufficiently productive, the effort gains from meritocracy render disarmament unnecessary.
Regardless of the distribution, the student-optimal mechanism does not involve exclusion.

Several extensions merit further investigation.
Competition among sellers of positional goods remains unexplored.
In the education context, this can be interpreted as competition between different tracks leading to different status levels (e.g., vocational versus academic education).
The problem of sellers competing strategically in mechanisms remains open.

\newpage
\counterwithin*{equation}{section}
\renewcommand{\theequation}{\thesection.\arabic{equation}}

\counterwithin*{theorem}{section}
\renewcommand{\thetheorem}{\thesection.\arabic{theorem}}

\counterwithin*{remark}{section}
\renewcommand{\theremark}{\thesection.\arabic{remark}}

\counterwithin*{lemma}{section}
\renewcommand{\thelemma}{\thesection.\arabic{lemma}}

\begin{appendices}
\section{Optimal Mechanisms without Exclusion} 
\label{app:noexclusion}
In this section, I study the optimal mechanisms when exclusion is impossible and the lowest-level positional good is available for free.%
\footnote{
This section subsumes results in my note ``A Mechanism Design Approach to `Gainers and Losers in Priority Services'.''
}
For example, in \cite*{MoldovanuSelaShi2007}, agents cannot be excluded and are guaranteed at least the lowest status in the organization; in the benchmark model of \citet[henceforth GW]{GershkovWinter2023}, buyers can use the regular (non-priority) service for free.
The free lowest-level position, rather than opting out, thus serves as the buyer's effective outside option.
Moreover, as noted in Remark~\ref{MPS(F)}, the feasibility condition becomes $s\in \MPS(F)$.

Based on the feasibility condition, I first characterize the revenue-maximizing mechanism (Proposition~\ref{prop8}).
In particular, I provide necessary and sufficient conditions under which the revenue-maximizing mechanism is fully separating.
I also establish a 2-approximation result: selling a premium tier alongside the free tier guarantees at least half the maximum revenue (Proposition~\ref{prop:approx2}).
Then, I study consumer surplus and provide conditions under which it is decreasing or increasing in the number of positional good levels offered by the seller (Proposition~\ref{prop7}).
I also characterize the necessary and sufficient condition under which introducing a premium tier alongside the free tier increases consumer welfare (Proposition~\ref{prop1}) and under which introducing premium tiers hurts every consumer (Proposition~\ref{prop2}).

\subsection{Revenue Maximization}
The revenue maximization problem is given by
\begin{align}
    \max_{s(\theta), p(\theta)} \int_{0}^{\bar \theta} p(\theta) \d F(\theta) 
\end{align}
subject to the following constraints on $\theta\in[0,\bar\theta]$
\begin{align}\label{constraints}
    & U(\theta) - v(\theta) = \theta s(\theta) - p(\theta) \geq \theta \lbar s &&\mbox{(IR)}\\
    & U(\theta) - v(\theta) =  U(0) - v(0) + \int_0^\theta s(t)  \d {t}  &&\mbox{(IC)}\\
    & s(\theta) \mbox{ is increasing} \\
    & s \in \MPS(F ) &&\mbox{(MPS)}\label{eqn:MPS}
\end{align}
where $\lbar s = \min\{s(\theta)\}$ denotes the lowest status, which is endogenously determined by the status allocation.
The (IR) constraint, which is equivalent to $\tilde U(\theta) \equiv U(\theta)-v(\theta)-\theta\lbar s\geq0$, can be reduced to $U(0)-v(0)\geq0$ because $U'(\theta)-v'(\theta)-\lbar s=s(\theta)-\lbar s\geq0$.

By standard arguments, the expected revenue is given by
\begin{align}
    R = \int_{0}^{\bar \theta} p(\theta) \d F(\theta) & = 
    \int_{0}^{\bar \theta} \left( \theta - \frac{1-F(\theta)}{f(\theta)} \right) s(\theta) \d F(\theta) - (U(0) - v(0)).
\end{align}
By (IR), it is optimal to set $U(0)=v(0)$.
The revenue maximization problem is equivalent to
\begin{equation}
    \max_{s \in \MPS(F)} \int_{0}^{\bar \theta}  J(\theta) s(\theta)  \d F(\theta).
\end{equation}
Theorem 4 (Fan--Lorentz) in KMS implies the following result.

\begin{proposition} \label{prop8}
    Offering more positional good levels always (strictly) increases the seller's revenue if and only if $J(\theta)$ is (strictly) increasing. 
    Thus, full separation maximizes revenue if and only if $J(\theta)$ is increasing.

    If $J(\theta)$ is not monotonic, define $\tilde{J}(\theta) = \int_0^\theta J(t) \d F(t)$ and let $K(\tau) = \conv \tilde J(F^{-1}(\tau))$ denote the convex hull of $\tilde J$ in quantile space.
    Then, the revenue-maximizing mechanism separates types if $K\circ F = \tilde J$ and pools types otherwise.
\end{proposition}
\begin{remark}
    GW's Proposition 8 provides a sufficient condition for the seller's revenue to be strictly increasing in the number of priority classes: $F$ satisfies the IFR property.%
\end{remark}

\begin{remark}
    Effort maximization in MSS's model is the same as revenue maximization because they assume linear effort costs (and no exclusion). In Theorem 4, they also provide a sufficient condition for full separation to be optimal: $F$ satisfies the IFR property.
\end{remark}

Infinitely many positional good levels can be implemented by an all-pay auction, in which the more money a consumer pays, the higher status he receives.

Now consider the approximation of selling a single positional good in addition to a free low-level position, as in Example~\ref{ex:twolevels}.
\begin{proposition} \label{prop:approx2}
   The seller can obtain at least half the maximum revenue by selling a single positional good in addition to a free low-level position.
\end{proposition}
\begin{proof}
    Denote by $p$ the price of the high-level position.
    The cutoff type $\theta(p)$ indifferent between two levels is given by \[\theta(p) \frac{1+F (\theta(p))}{2} -p =\theta(p)\frac{F(\theta(p))}{2} \iff \theta(p)=2p.\]
    Let $R_2$ denote the seller's maximum revenue from offering one positional good in addition to a free low-level position, which is given by
    \[R_2 = \max_p p(1-F(2p)) = \frac{1}{2}\max_p p(1-F(p)).\]
    Consider the auxiliary problem of selling an indivisible good to one buyer, in which a standard extreme-point argument implies a posted-price mechanism is optimal \citep[see, e.g.,][Proposition 2.5]{Borgers2015}. 
    Formally, let $\mathcal M= \{q\colon [0,\bar \theta ]\to [0,1] \mid q \text{ increasing}\}$ denote the set of incentive-compatible allocations, then 
    $$\max_{q\in\mathcal M} \int_0^{\bar \theta} J(\theta) q(\theta) \d F(\theta) = \max_p p(1-F(p)).$$
    Any maximizer $\tilde q^*$ must satisfy $\tilde q^*(\theta)=1$ on $[\tilde p^*,\bar\theta]$ for some $\tilde p^*<\bar\theta$.
    Because $\MPS(F)\subseteq \mathcal M$, we have
    \[
    \begin{aligned}
        \max R  = \max_{s \in \MPS(F)} \int_{0}^{\bar \theta}  J(\theta) s(\theta)  \d F(\theta)
        < \max_{{x}\in \mathcal M} \int_0^{\bar \theta} J(\theta) {x}(\theta) \d F(\theta)
         = \max_p p(1-F(p)) = 2R_2.
    \end{aligned}   
    \]
    The inequality is strict because any maximizer of the auxiliary problem assigns $q(\theta)=1$ to a positive measure of types, which is not in the feasible set $\MPS(F)$.
\end{proof}

This proposition extends Proposition~\ref{prop:approx} to the case where exclusion is impossible.

\paragraph{Graphical Illustration.}
Figure~\ref{fig:approx1} plots the \emph{revenue curve} $R(\tau)=(1-\tau)F^{-1}(\tau)$ in quantile space, where $\tau=F(\theta)$.
Assume $v(\theta)=0$ for simplicity.
Similar to Section~\ref{sec:approx}, 
(i) the area under the revenue curve (\emph{blue} area in the left panel) represents the revenue from the revenue-maximizing mechanism---i.e., $s^*(\theta) = F(\theta)$,
(ii) the \emph{red} triangle in the right panel represents the revenue from having two levels of positional goods---a higher level offered at price $\theta^*_0/2$ in addition to a lower level offered for free, %
and (iii) the entire box of either panel (which has an area of $R(\tau^*_0)$) represents the maximum revenue from selling nonpositional goods in the auxiliary problem.
It is straightforward that the red area (ii) is at least half of the entire box and thus strictly larger than half of the blue area (i).

\begin{figure}[hbt]
    \centering
    \begin{subfigure}[h]{0.45\textwidth}
        \centering
        \includegraphics[width= \textwidth]{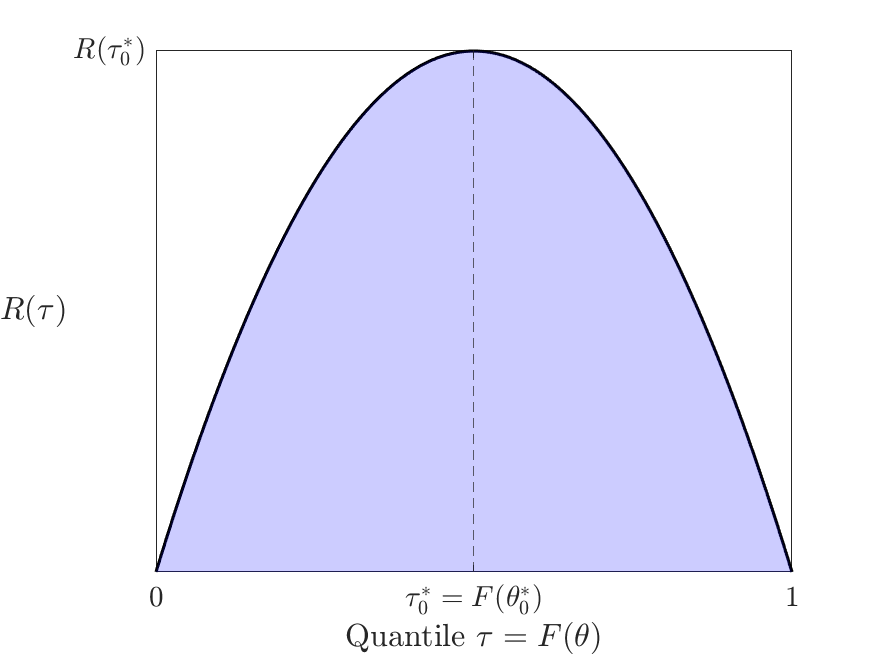}
    \end{subfigure}
    \begin{subfigure}[h]{0.45\textwidth}
        \centering
        \includegraphics[width= \textwidth]{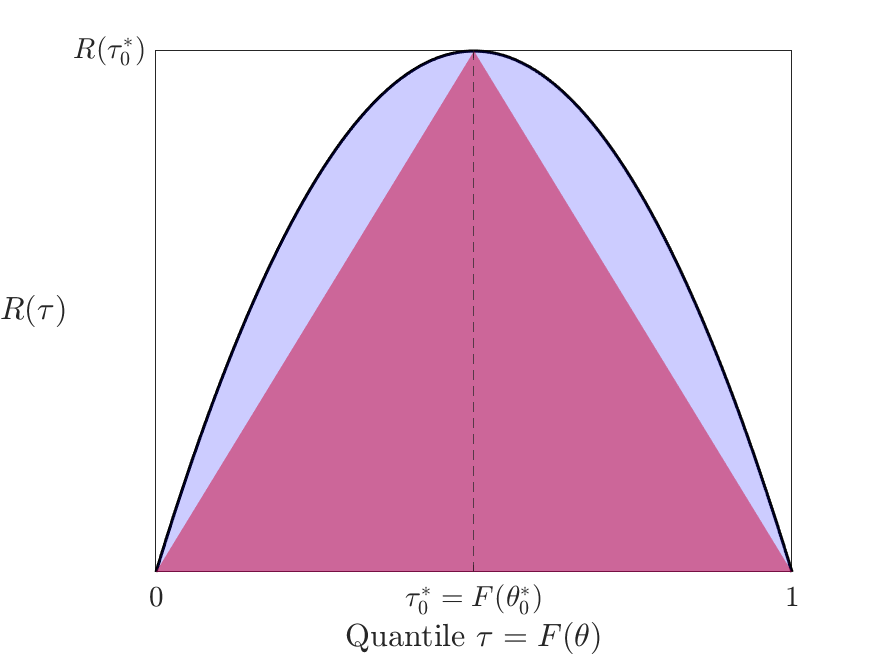}
    \end{subfigure}
    \caption{Revenue curve $R(\tau)$ for the uniform distribution}
    \label{fig:approx1}
\end{figure} 

The approximation performs well for many common distributions, as shown in the following examples (with $v(\theta)=0$).

\begin{example*}[Exponential distribution]
    If $F(\theta)=1-\exp(-\lambda\theta)$ where $\lambda>0$, offering a single good above the free lowest level can obtain $73.6\%$ of the maximum revenue.
\end{example*}
\begin{example*}[Uniform distribution]
    If $F(\theta)=\theta$ on $[0,1]$, 
    offering a single good above the free lowest level can obtain $75\%$ of the maximum revenue.
\end{example*}
\begin{example*}[Power distribution]
    If $F(\theta)=\theta^\beta$ on $[0,1]$ where $\beta>0$, offering a single good above the free lowest level can obtain 
    \begin{equation*}
        \frac{R_2}{\max R} =
        \begin{cases}
          \displaystyle  \frac{(1+2\beta)}{2\beta (1+\beta)^{1/\beta}}    > 72.1\%, &\text{ if }\beta\geq1\\
          \displaystyle  \frac{1+2\beta}
            {
            (1+\beta)^{1/\beta}
            \left[
            2\beta+(1-\beta)^{2+1/\beta}
            \right]
            } > \frac{3}{4}, &\text{ if }0<\beta<1
        \end{cases}
    \end{equation*}
    of the maximum revenue. The ratio approaches 1 as $\beta\to \infty$ (i.e., $F$ is sufficiently convex) or $\beta\to 0$ (i.e., $F$ is sufficiently concave).
\end{example*}

\subsection{Consumer Surplus}
\label{noexclusion:consumer}
Consumer surplus is given by
\begin{equation}
    W = \int_{0}^{\bar \theta} U(\theta) \d F(\theta)  = \int_{0}^{\bar \theta}  \left(\frac{1-F(\theta)}{f(\theta)} \right)  s(\theta)   \d F(\theta) + \E[v(\theta)] +  U(0) - v(0).
\end{equation}

\begin{proposition}\label{prop7}
    For any increasing status profiles $s,\hat{s}\colon [0,\bar\theta]\to [0,1]$ such that $\hat{s}\in \MPS(s)$,
    \begin{enumerate}[label=(\roman*)]
        \item $\hat{s}$ results in a higher consumer surplus if $F$ satisfies IFR;
        \item $\hat{s}$ results in a lower consumer surplus if $F$ satisfies DFR.
    \end{enumerate} 
    Therefore, consumer surplus decreases (increases) as the seller offers more levels of positional goods if $F$ satisfies IFR (DFR). 
\end{proposition}

I focus on the case where the price (or effort) is nonnegative, which pins down $p(0)= v(0)-U(0)=0$.

\begin{corollary}
        Consumer surplus is the highest when the seller offers
        \begin{enumerate}[label=(\roman*)]
            \item one position level (i.e., $s(\theta)=1/2$) if $F$ satisfies IFR;
            \item full separation (i.e., $s(\theta)=F(\theta)$) if and only if  $F$ satisfies DFR.
        \end{enumerate} 
    In the former case, consumer surplus is $\E[v(\theta)+\theta/2]$.
\end{corollary}

\begin{remark}
    GW's Proposition 7 shows the sufficiency of the IFR property ($\frac{1-F(\theta)}{f(\theta)}$ is decreasing) for no priority service to be consumer surplus-maximizing.
\end{remark}

\begin{remark}
    If both $J(\theta) = \theta - \frac{1-F(\theta)}{f(\theta)}$ and $\frac{1-F(\theta)}{f(\theta)}$ are increasing (e.g., exponential and Pareto distributions), a finer partition (in terms of majorization) can increase both the seller's revenue and consumer surplus.
    Thus, full separation maximizes both the revenue and consumer surplus.
\end{remark}

\begin{proposition} \label{prop1}
  Consumer surplus is the highest when the seller offers a single (free) level if and only if 
  \[
  \int_0^\theta \left( \frac{1-F(t)}{f(t)} -\E[\theta] \right)\d F(t) \geq0 \quad \mbox{ for all $\theta\in[0,\bar \theta]$.}
  \]
 \end{proposition}
 \begin{proof}
     Define $H(\theta) = \int_0^\theta \frac{1-F(t)}{f(t)} \d F(t)$.
     Then, the condition in the proposition is equivalent to $H(\theta) \geq H(\bar \theta) F(\theta)=\E[\theta]F(\theta)$ (i.e., $H(\theta)$ lies above the line connecting $H(0)=0$ and $H(\bar \theta)=\E[\theta]$ in quantile space).
      Therefore, by Proposition 2 in KMS, this condition is necessary and sufficient for total pooling to be welfare-maximizing.
 \end{proof}

 \begin{remark}
    A sufficient condition is that $\frac{1-F(\theta)}{f(\theta)}$ single-crosses $\E[\theta]$ from above.
    An even stronger sufficient condition is the IFR property, i.e., $\frac{f(\theta)}{1-F(\theta)}$ is increasing.%
 \end{remark}

 \begin{remark}
    This condition is necessary and sufficient for customers' welfare to be higher when a free level is offered than when \emph{any} $k>1$ levels are offered. 
    By contrast, GW's Proposition 1 provides a sufficient condition for customers' welfare to be higher when a free level is offered than when one additional level is introduced.
 \end{remark}

Moreover, fix a type $\theta\in [0,\bar\theta]$, maximizing his utility $U(\theta)$ subject to constraints~\eqref{constraints}--\eqref{eqn:MPS} is equivalent to
\begin{equation}
    \max_{s \in \MPS(F)} \int_0^{\theta} \frac{s(t)}{f(t)} \d F(t) + v(\theta) %
\end{equation}
The following proposition provides a necessary and sufficient condition for total pooling to maximize $U(\theta)$ for all $\theta\in [0,\bar \theta]$.
 
\begin{proposition} \label{prop2}
    Every consumer's utility is the highest when the seller offers one position level if and only if $F(\theta)\leq \theta/\bar \theta$ for all $\theta\in[0,\bar\theta]$ (i.e., $F$ first-order stochastic dominates the uniform distribution). A sufficient condition is that $f(\theta)$ is increasing.
\end{proposition}
\begin{proofsketch}
    $U(\theta) = \int_0^{\theta} \frac{s(t)}{f(t)} \d F(t) + v(\theta)$.
    If $f$ is increasing, then $\frac{1}{f(\theta)}$ is decreasing, so $s=1/2$ maximizes $U(\theta)$.  

    Note that $\int_0^{\theta} \frac{1}{f(t)} \d F(t) = \theta$, and the condition $F(\theta)\leq \theta/\bar \theta$ is equivalent to $\int_0^\theta (\frac{1}{f(t)} - \bar \theta) \d F(t) \geq 0$ for all $\theta\in[0,\bar\theta]$.
    By the same argument as in the proof of Proposition~\ref{prop1}, this condition is necessary and sufficient for $s(\theta)=1/2$ to maximize $U(\theta)$ for all $\theta \in [0,\bar\theta]$.
\end{proofsketch}
\begin{remark}
    Proposition~\ref{prop2} provides the necessary and sufficient condition for all consumers to be worse off after \emph{any} $k>1$ levels of positional goods are offered than if one level is offered (for free). 
    GW's Proposition 2 shows that if $F(c)\leq c/\bar c$, all consumers are worse off after the introduction of \emph{one} additional level of priority service.
\end{remark}

\section{Proofs}
\label{sec:proofs}

\subsection{Proof of Theorem~\ref{thm:feasibility}}
\begin{proof}
    ($\implies$) 
Suppose that an increasing status profile $s$ is induced by some allocation
$\chi\colon\Theta\times\Omega\to\mathbb R_+$.
By Lemma~\ref{lemma:IC}, there exists a participation cutoff
$\theta_0\in[0,\bar\theta]$ such that types below $\theta_0$ do not
participate in any realization, whereas types above $\theta_0$
participate in every realization. Hence,
$s(\theta)=0$ for all $\theta<\theta_0$.

Fix $t\in[\theta_0,\bar\theta]$. For each realization $\omega$, define the
realized status of type $\theta$ by
\[
S_\omega(\theta)
:=
S\bigl(\chi(\theta,\omega),G_\chi(\cdot\mid\omega)\bigr).
\]
Let $S_\omega^\uparrow\colon[\theta_0,\bar\theta]\to[0,1]$ denote the
increasing rearrangement of $S_\omega$ among participants. The rearrangement inequality gives
\[
\int_t^{\bar\theta}S_\omega(\theta) \d F(\theta)
\leq
\int_t^{\bar\theta}S_\omega^\uparrow(\theta) \d F(\theta).
\]
Moreover, in each realized tier, quantile ranks are replaced by their
average. Hence, the full-separation profile $F(\theta)$ majorizes
$S_\omega^\uparrow(\theta)$, so
\[
\int_t^{\bar\theta}S_\omega^\uparrow(\theta) \d F(\theta)
\leq
\int_t^{\bar\theta}F(\theta) \d F(\theta).
\]
Taking expectations over $\omega$ therefore yields
\[
\int_t^{\bar\theta}s(\theta) \d F(\theta)
\leq
\int_t^{\bar\theta}F(\theta) \d F(\theta),
\qquad
\forall t\in[\theta_0,\bar\theta].
\]
Finally, each realized tier preserves total status because it replaces
the quantile ranks assigned to that tier by their average. Summing over
all realized tiers and taking expectations over $\omega$ gives
\[
\int_{\theta_0}^{\bar\theta}s(\theta) \d F(\theta)
=
\int_{\theta_0}^{\bar\theta}F(\theta) \d F(\theta).
\]

    ($\impliedby$) Suppose there exists $\theta_0\in[0,\bar\theta]$ such that $s \in\MPS(F\cdot\mathbf{1}_{[\theta_0,\bar\theta]})$ and $s(\theta)=0$ for all $\theta<\theta_0$.
    I show that $s$ is feasible.

    First, by Theorem 1 in KMS, any extreme point $\tilde s$ of $\MPS(F\cdot\mathbf{1}_{[\theta_0,\bar\theta]})$ such that $\tilde s(\theta)=0$ for all $\theta<\theta_0$ is characterized by a collection of intervals $[\underline\theta_i,\bar\theta_i)\subseteq [\theta_0,\bar\theta]$, indexed by $i\in I$, such that 
    \[
        \tilde s(\theta)
        =
        \frac{1}{F(\bar\theta_i)-F(\underline\theta_i)}
        \int_{\underline\theta_i}^{\bar\theta_i} F(\tilde\theta)\,dF(\tilde\theta)
        =
        \frac{F(\underline\theta_i)+F(\bar\theta_i)}{2},
        \quad \text{ if }
        \theta\in[\underline\theta_i,\bar\theta_i],
    \]
    and $\tilde s(\theta)=F(\theta)$ if 
        $\theta\notin\bigcup_{i\in I} [\underline\theta_i,\bar\theta_i)$.

    Each such extreme point is implementable by an increasing deterministic allocation $\chi\colon \Theta \to \reals_+$. 
    For all $\theta<\theta_0$, assign the outside option $\chi(\theta)=0$. 
    For all $\theta\geq\theta_0$, choose a (weakly) increasing $\chi$ such that: on each pooling interval $[\underline\theta_i,\bar\theta_i)$, $\chi(\theta)=x_i>0$; outside the pooling intervals, assign a strictly increasing $\chi(\theta)$.  
    Under the status specification 
    \[ S(x,G)= 
    \begin{cases}
        \dfrac{G^-(x)+G(x)}{2} & \text{if } x>0,\\
        0 & \text{if } x=0,
    \end{cases} \] 
    this deterministic allocation induces the extreme point $\tilde s$.
    Indeed, if $\chi$ is strictly increasing at $\theta$, then $G_\chi(\chi(\theta))
=
G_\chi^-(\chi(\theta))
=
F(\theta)$.
    If $\theta\in[\underline\theta_i,\bar\theta_i)$, then $G_\chi(\chi(\theta))=F(\bar\theta_i)$ and $G_\chi^-(\chi(\theta))=F(\underline\theta_i)$.
    
    Finally, I show that $s$ is feasible.
    By Proposition 1 in KMS, there exists a probability measure $\lambda_s$ supported on the extreme points of $\MPS(F\cdot\mathbf{1}_{[\theta_0,\bar\theta]})$, denoted by $\mathcal{E}$, such that $s=\E[\tilde s\mid 
    \tilde s\sim \lambda_s]$ for all $\theta\in[\theta_0,\bar\theta]$.    
    Since $\omega\sim \text{Unif}[0,1]$, there exists a measurable function $\tilde{S}\colon [0,1]\to \mathcal{E}$ such that $\tilde{S}(\omega)\sim \lambda_s$.  
    In other words, the random variable $\omega$ draws an extreme point $\tilde{S}(\omega)$ of $\MPS(F\cdot\mathbf{1}_{[\theta_0,\bar\theta]})$ according to $\lambda_s$.

    For each extreme point $\tilde s$, let $\chi_{\tilde s}$ be a deterministic allocation that induces $\tilde s$. 
     Define the stochastic allocation $\chi(\theta,\omega) = \chi_{\tilde{S}(\omega)}(\theta)$.
     Conditional on the realization of $\omega$, the allocation $\chi(\cdot,\omega)$ induces $\tilde{S}(\omega)$. Therefore, for every $\theta\in[\theta_0,\bar\theta]$,
     \[ \E_\omega  [ S (\chi(\theta,\omega),G_\chi(\cdot\mid\omega) )  ]  = \E_\omega[\tilde{S}(\omega)(\theta)] = \int_{\mathcal{E}} \tilde s(\theta)\d\lambda_s(\tilde s) = s(\theta). \] 
    For $\theta<\theta_0$, we have $s(\theta)=0$ is induced by $\chi(\theta,\omega)=0$ because $S(0,G_\chi)=0$.
    Thus, $s$ is feasible.
\end{proof}

\subsection{Proofs of Section~\ref{exclusion}}
\label{app:proofs_2}
\subsubsection{Proofs of Section~\ref{sec:revenue}}

\begin{proof}[Proof of Proposition~\ref{prop:revmax}]
    Because it is optimal for the seller to set $U(\theta_0)=0$ and $p(\theta_0) = \theta_0 s(\theta_0) + v(\theta_0)$, revenue reduces to
    \begin{equation} \label{eqn:revmax}
        R(\theta_0) =  \int_{\theta_0}^{\bar \theta} \left( J(\theta) s(\theta)  + v(\theta) - \frac{1-F(\theta)}{f(\theta)} v'(\theta) \right) \d F(\theta).
    \end{equation}
    The following lemma follows from applying Theorem 4 (Fan--Lorentz) in KMS to $R(\theta_0)$.
    \begin{lemma}\label{lem:fanlorentz}
        Assume $J(\theta)$ is increasing. For any two incentive-compatible mechanisms $(s_1(\theta),p_1(\theta))$ and $( {s}_2(\theta),{p}_2(\theta))$ with the same cutoff $\theta_0$ and $U(\theta_0)=0$, if $s_1\succ s_2$ (in the majorization order), then $(s_1(\theta),p_1(\theta))$ generates higher revenue than $({s}_2(\theta),{p}_2(\theta))$.
    \end{lemma}

    As the seller offers more levels of positional goods, it is feasible to offer a finer allocation (in terms of majorization) with the same exclusion level $\theta_0$, which by Lemma~\ref{lem:fanlorentz}, increases revenue.

    Thus, the revenue-maximizing mechanism fully separates the participants (i.e., $s(\theta)=F(\theta)\cdot \mathbf{1}_{\theta\geq \theta_0^*}$) and the optimal cutoff type $\theta_0^*$ is given by
    \begin{align*}
        \theta_0^* 
        \in \argmax_{\theta_0}  \int_{\theta_0}^{\bar \theta} J(\theta) F(\theta) \d F(\theta) + v(\theta_0)(1-F(\theta_0)).
    \end{align*}
    
    If $J(\theta)$ is not monotonic, define $\tilde{J}(\theta;\theta_0) = \int_{\theta_0}^\theta J(t) \d F(t)$ and apply the ironing technique \citep{Toikka2011,KleinerMoldovanuStrack2021}.
    Let $$K(\tau) = \conv \tilde J(F^{-1}(\tau))$$ denote its convex hull on $[F(\theta_0),1]$ in quantile space $\tau = F(\theta)$.
    Then, the revenue-maximizing mechanism excludes types below $\theta^*_0$ and 
    pools types into the same positional good level if $K$ is affine and separates types if $K\circ F = \tilde J$.
    The optimal cutoff type is given by $\theta^*_0 \in \argmax_{\theta_0}  \int_{\theta_0}^{\bar \theta} F(\theta) \d K(F(\theta)) + v(\theta_0)(1-F(\theta_0))$.
\end{proof}

\begin{proof}[Proof of Corollary~\ref{cor:exclusion}]

    The following lemma shows that $J_v(\theta)$ is increasing if $J(\theta)$ is increasing under the assumptions on $v(\theta)$.
    \begin{lemma} \label{lem:Jv}
        If $J(\theta)$ is increasing, \(v'(\theta)\geq 0\), and \(v''(\theta)\le 0\), then
        $J_v(\theta)=v(\theta)-\frac{1-F(\theta)}{f(\theta)}v'(\theta)$
        is increasing.
    \end{lemma}
    \begin{proof}
    Define the inverse hazard rate of $F$ as $
    h(\theta)=\frac{1-F(\theta)}{f(\theta)}.$
    Thus,
    \[
    \begin{aligned}
    J_v'(\theta)
    &=
    v'(\theta)-h'(\theta)v'(\theta)-h(\theta)v''(\theta) \\
    &=
    J'(\theta)v'(\theta)-h(\theta)v''(\theta)\geq0
    \end{aligned}
    \]
    because \(J'(\theta)\ge 0\), \(v'(\theta)\geq 0\), \(h(\theta)>0\), and \(v''(\theta)\le 0\).
    Hence, $J_v(\theta)$ is increasing. 
    \end{proof}
    
    Define the marginal revenue function as $\psi(\theta)= J(\theta)
    F(\theta) + J_v(\theta)$.
    Let $\theta_1 = \inf\{\theta\in[0,\bar\theta]\colon J(\theta)> 0\}$ denote the optimal cutoff type in the absence of intrinsic value (and hence $J_v(\theta) \equiv 0$), which satisfies $J(\theta_1) = 0$.
    We have 
    \begin{equation}
        \psi(\theta_1) = J(\theta_1) F(\theta_1)  + J_v(\theta_1) = J_v(\theta_1) = v(\theta_1) - \theta_1 v'(\theta_1) \geq v(0) \geq 0,
    \end{equation}
    by the concavity of $v(\theta)$.
    Moreover, because $J(\theta)$ and $J_v(\theta)$ are increasing, we have
    \begin{equation}
        \psi'(\theta) = J'(\theta)F(\theta) + J(\theta)f(\theta) + J_v'(\theta) > 0 \mbox{  for all $\theta>\theta_1$.}
    \end{equation}
    Thus, we have $\psi(\theta)> 0$ and thus $R'(\theta)< 0$ for all $\theta>\theta_1$.
    Therefore, the optimal cutoff satisfies $\theta^*_0 \leq \theta_1$.

    In particular, if $v(\theta)=\alpha\theta$ for some $\alpha\geq0$, then $J_v(\theta) = \alpha J(\theta)$, so $\psi(\theta) = (F(\theta)+\alpha) J(\theta)$ has the same sign as $J(\theta)$. Thus, $\theta_0^* = \theta_1$.

    (ii) Because $J(\theta)$ is increasing, $J(\theta)F(\theta) \geq J(0) F(\bar\theta)= -1/f(0)$.
    Thus, if $J_v(0)-1/f(0) = v(0)- (v'(0)+1)/f(0) \geq 0$, the integrand is always nonnegative, so zero exclusion is optimal (i.e., $\theta_0^* =0$).
\end{proof}

\subsubsection{Proofs of Section~\ref{sec:approx}}

\begin{proof}[Proof of Corollary~\ref{cor:sep}]
Fix a participation cutoff $\theta_0$.
In quantile space, let $I(\tau)=\tilde J(F^{-1}(\tau)) = \int_{\theta_0}^{F^{-1}(\tau)} J(\theta) \d F(\theta)$ for all $\tau\in[F(\theta_0),1]$.
The following lemma establishes a useful property of chords ending at the upper endpoint.
\begin{lemma}\label{lem:chord}
    For all $\tau\in[F(\theta_0),1)$, the chord joining $(\tau,I(\tau))$ and $(1,I(1))$ has slope $F^{-1}(\tau)$ strictly increasing in $\tau$.
\end{lemma}

\begin{proof}
    For any $\theta<\bar\theta$,   we have
\begin{align}
\tilde J(\bar\theta)-\tilde J(\theta)=
\int_{\theta}^{\bar\theta}
\left(
t-\frac{1-F(t)}{f(t)}
\right)
 \d F(t) =
\theta(1-F(\theta)).
\label{eq:terminal-chord}
\end{align}
If $\bar\theta=\infty$, the integration-by-parts argument uses
$\lim_{\theta\to\infty}\theta(1-F(\theta))=0$, which follows from the
finite-mean assumption. 

Therefore, in quantile space, the slope of the chord joining $(\tau,I(\tau))$ and $(1,I(1))$ is
\[ \frac{I(1)-I(\tau)}{1-\tau} = F^{-1}(\tau)\]
which is strictly increasing in $\tau$.
\end{proof}

Based on Lemma~\ref{lem:chord},
I first show it is never optimal to pool a nondegenerate interval of highest types.
Suppose, toward contradiction, that there exists a maximal pooling interval of highest types
$[\hat\theta,\bar\theta]$, where $\hat\theta<\bar\theta$. Then, 
by Proposition~\ref{prop:revmax}, $K(\tau)$ is affine on $[F(\hat\theta),1]$ and coincides with $I(\tau)$ at both endpoints: $
K(1)=I(1)$ and $K(F(\hat\theta))=I(F(\hat\theta))$.
Hence, the slope of $K$ on $[F(\hat\theta),1]$ is
\[
\frac{I(1)-I(F(\hat\theta))}{1-F(\hat\theta)}
= 
\hat\theta.
\]
Now take any $\hat\theta'\in(\hat\theta,\bar\theta)$ and let
$\tau'=F(\hat\theta')$. Since $K$ is affine with slope $\hat\theta$ and
$K(1)=I(1)$, we have
\[
K(\tau')
=
I(1)-\hat\theta(1-\tau')< I(1)-\hat\theta'(1-\tau')=I(\tau')
\]
The strict inequality follows from $\hat\theta'>\hat\theta$.
This contradicts the fact that $K$ is the lower convex hull of $I$. Hence, the optimal mechanism cannot pool
a nondegenerate interval of highest types.

Let $R_n$ denote the maximum revenue when the seller is restricted to at most $n$ tiers.
Consider an optimal $n$-tier mechanism. By Lemma~\ref{lem:chord}, splitting its highest tier into two replaces the chord ending at $\tau=1$ with two chords that lie strictly below the original chord, thereby strictly increasing revenue. Therefore, $R_n < R_{n+1}$ for all $n\geq 1$.
\end{proof}

\begin{proof}[Proof of Proposition~\ref{prop:approx}]
    Denote by $p$ the price of the positional good.
    The cutoff type $\theta(p)$ indifferent between buying and not buying is given by 
    \[\theta(p)  \frac{1+F (\theta(p))}{2}   -p + v(\theta(p)) = 0 \implies p  = \theta(p)   \frac{1+F (\theta(p))}{2}   + v(\theta(p)).\]
    Because $v'(\theta)\geq0$, all types above $\theta(p)$ will buy the good.
    Thus,  by selling a single good, the seller's {maximum} revenue is
    \begin{equation} \label{eqn:R1}
        R_1 =  \max_\theta    \left( \theta \frac{1+F (\theta )}{2} + v(\theta ) \right) (1-F(\theta)) \geq  \frac{1}{2}\max_\theta (\theta + v(\theta))(1-F(\theta)) 
    \end{equation}
    because $\frac{1+F (\theta )}{2}\theta\geq \frac{1}{2}\theta$.
    Assume first for simplicity $v(\theta)=0$, so we have 
    $R_1  \geq  \frac{1}{2}\max_\theta \theta(1-F(\theta))$.
    
    Consider the auxiliary problem of selling an indivisible good to one buyer, in which a standard extreme-point argument implies a posted-price mechanism is optimal \citep[see, e.g.,][Proposition 2.5]{Borgers2015}. 
    Formally, let $q$ denote the probability that the agent receives the good, and let $\mathcal M= \{q\colon [0,\bar \theta ]\to [0,1] \mid q \text{ increasing}\}$ denote the set of incentive-compatible allocations.
    The agent's payoff is $u(p,q, \theta) = \theta q - p$.
    Then, the standard argument implies that 
    $$\max_{q\in\mathcal M} \int_0^{\bar \theta} J(\theta) q(\theta) \d F(\theta) = \max_p p(1-F(p)),$$ 
    and a posted-price mechanism $q^*(\theta) = \mathbf1_{\theta\geq p}$, where $p^*\in \argmax_p p(1-F(p))$, which is an extreme point of $\mathcal M$, is optimal.
    Further, any optimal mechanism $\tilde q$ must satisfy $\tilde q(\theta)=1$ on $[\tilde p,\bar\theta]$ for some $\tilde p<\bar\theta$

    Let $R^*$ denote the maximum revenue from selling positional goods.
    Because $\MPS_0(F) \subseteq {\mathcal M}$, we have
    \[  R^*  = \max_{s \in \MPS_0(F)} \int_{0}^{\bar \theta}  J(\theta) s(\theta) \d F(\theta) < \max_{q\in {\mathcal  M}} \int_0^{\bar \theta} J(\theta) q(\theta) \d F(\theta)
    = \max_p p(1-F(p)) \leq 2R_1.\]
    The inequality is strict because any maximizer of the auxiliary problem assigns $q(\theta)=1$ to a positive measure of types, which is not in the feasible set $\MPS_0(F)$.

    Now consider the general case where $v(\theta)\geq0$ and $v'(\theta)\geq0$. 
    I show that the same bound continues to hold.
    We have
    \begin{align*}
        R^*
        &= \max_{s \in \MPS(F\cdot\mathbf{1}_{[\theta_0,\bar\theta]}),\theta_0} \left[ \int_{\theta_0}^{\bar \theta} J(\theta) s(\theta)   \d F(\theta) + v(\theta_0)(1-F(\theta_0)) \right] \\
        &<\max_{\theta_0} \max_{q\in {\mathcal  M}}\left[ \int_{\theta_0}^{\bar \theta}  J(\theta) q(\theta) \d F(\theta) + v(\theta_0)(1-F(\theta_0)) \right] \\
        &\leq \max_{\theta_0} \left[ \max_{\theta\geq \theta_0} \theta (1-F(\theta)) + v(\theta_0)(1-F(\theta_0)) \right].
    \end{align*}
    The first strict inequality follows from the fact that every maximizer of the relaxed problem assigns $q(\theta)=1$ to a positive measure of types, which is not feasible in the original problem.
    The second inequality follows from the optimality of the posted-price mechanism.

    In order to establish that $R^* < 2 R_1$, it remains to show the following claim:
    \begin{claim*}
    Fix any $\theta_0$, we have $2 R_1 \geq \theta (1-F(\theta)) +
    v(\theta_0)(1-F(\theta_0))$ for all $\theta\geq \theta_0$.
    \end{claim*}
    
    \begin{proof}[Proof of the claim]
    For all $\theta\geq \theta_0$, because $v(\theta)$ is increasing, we have $v(\theta)\geq v(\theta_0)$ and thus
        \[
        R_1
        \geq
        \left(
        \theta\frac{1+F(\theta)}2
        +
        v(\theta)
        \right)(1-F(\theta))
        \geq
        \left(
        \theta\frac{1+F(\theta)}2
        +
        v(\theta_0)
        \right)(1-F(\theta)).
        \]
    Also, 
    \[
     R_1
    \geq
    \left(
    \theta_0\frac{1+F(\theta_0)}2
    +
    v(\theta_0)
    \right)(1-F(\theta_0))
    \geq
    v(\theta_0)(1-F(\theta_0)).
    \]
    There are two cases.
    First, if $v(\theta_0)(1-F(\theta_0))\geq \theta (1-F(\theta))$, then the second inequality above implies
    $$2 R_1 \geq 2 v(\theta_0)(1-F(\theta_0)) \geq \theta (1-F(\theta)) + v(\theta_0)(1-F(\theta_0)).$$
    Otherwise, if $v(\theta_0)(1-F(\theta_0))< \theta (1-F(\theta))$, then
    \[
    0\leq v(\theta_0)
<
\frac{\theta(1-F(\theta))}{1-F(\theta_0)}.
    \]
It is enough to show that
\[
\left(
\theta\frac{1+F(\theta)}{2}
+
v(\theta_0)
\right)(1-F(\theta))
\geq
\frac{1}{2}
\left[
\theta(1-F(\theta))
+
v(\theta_0)(1-F(\theta_0))
\right],
\]
or equivalently,
\[
\frac{1}{2}
\left(
\theta F(\theta)(1-F(\theta))
+
v(\theta_0)
\left[
2(1-F(\theta))-(1-F(\theta_0))
\right]
\right)\geq0.
\]
This expression is affine in $v(\theta_0)$. Therefore, 
it is enough to check the two endpoints $v(\theta_0)=0$ and $v(\theta_0)=\frac{\theta(1-F(\theta))}{1-F(\theta_0)}$.
At $v(\theta_0)=0$, the expression equals
\[
\frac{1}{2}\theta F(\theta)(1-F(\theta))\geq 0.
\]
At $v(\theta_0) = \frac{\theta(1-F(\theta))}{1-F(\theta_0)}$,
the expression equals
\[
\begin{aligned}
&\frac{1}{2}
\left[
\theta F(\theta)(1-F(\theta))
+
\frac{\theta(1-F(\theta))}{1-F(\theta_0)}
\left[
2(1-F(\theta))-(1-F(\theta_0))
\right]
\right] 
=
\frac{1}{2}\theta(1-F(\theta))^2
\frac{1+F(\theta_0)}{1-F(\theta_0)}
\geq 0.
\end{aligned}
\]
Thus,
\[
\left(
\theta\frac{1+F(\theta)}{2}
+
v(\theta_0)
\right)(1-F(\theta))
\geq
\frac{1}{2}
\left[
\theta(1-F(\theta))
+
v(\theta_0)(1-F(\theta_0))
\right].
\]
Since $R_1$ is at least the left-hand side, we have
\[
2R_1
\geq
\theta(1-F(\theta))
+
v(\theta_0)(1-F(\theta_0)).
\]
Therefore, the desired inequality holds for every $\theta\geq\theta_0$.
\end{proof}
Hence, we have established that $ R_1 > \frac{1}{2} R^*$.
\end{proof}

\subsubsection{Proofs of Section~\ref{sec:welfare}}
\begin{proof}[Proof of Proposition~\ref{prop:consumer}]
    (i) 
    First, I show the following lemma.
    The lemma follows from applying Fan--Lorentz theorem to 
    \begin{align}
    W(\theta_0) =   \int_{\theta_0}^{\bar \theta} \left( \frac{1-F(\theta)}{f(\theta)} \right)  (s(\theta)+ v'(\theta))  + U(\theta_0)  \d F(\theta).
    \end{align}
    \begin{lemma}\label{lem:fanlorentz2}
        For any two incentive-compatible mechanisms $(s_1(\theta),p_1(\theta))$ and $( {s}_2(\theta),{p}_2(\theta))$ with the same cutoff $\theta_0$ and $U(\theta_0)=0$, if $s_1\succ s_2$, then $(s_1(\theta),p_1(\theta))$ generates lower (higher) consumer surplus than $({s}_2(\theta),{p}_2(\theta))$ if $F$ satisfies IFR (DFR).
    \end{lemma}

    The lemma implies that, fixing the exclusion level $\theta_0$, the consumer surplus decreases (increases) when the status profile becomes finer under IFR (DFR).

    By the same argument as in Corollary~\ref{cor:exclusion}, if $v(0) \geq \frac{v'(0)+1}{f(0)}$, the optimal exclusion level is $\theta_0^*=0$.
    Alternatively, we have $\theta_0^*=0$ if the basic tier is offered for free.

    (ii)
    When the number of levels is unconstrained, because $J(\theta)$ is increasing, the seller chooses $s^*(\theta)=F(\theta)$ for all $\theta\in[\theta_0,\bar\theta]$.
    Thus, consumer surplus is given by
\begin{align*} 
   W(\theta_0) 
   =  \int_{\theta_0}^{\bar \theta}  \frac{1-F(\theta)}{f(\theta)}  (F(\theta)+ v'(\theta))    \d F(\theta),
\end{align*}
which is decreasing in $\theta_0$.

    When the number of levels is constrained to a single level, I show a stronger result: holding the status profile of consumers of higher tiers than the basic tier fixed, decreasing the exclusion level $\theta_0$ increases consumer surplus under IFR.
    The single-tier case is covered because there are no tiers above the basic tier.

    Fix an allocation $s(\theta)$ and its exclusion level $\theta_0$.
    If $s(\theta)$ is \emph{strictly} increasing at $\theta_0$, then decreasing $\theta_0$ to $\theta_0-\varepsilon$ increases the consumer surplus for all $\theta\in[\theta_0-\varepsilon,\theta_0)$ and does not affect the status for all $\theta\in[\theta_0,\bar \theta]$.
    Hence, the consumer surplus increases.

    If $s(\theta)$ is constant in a neighborhood of $\theta_0$, there exists some $\theta_1\in(\theta_0,\bar\theta]$ such that $s(\theta)=s_0$ is constant on $[\theta_0,\theta_1]$ and $s(\theta)> s_0$ for $\theta>\theta_1$.
    Thus, consumer surplus is
    \begin{equation*}
        W(\theta_0)  =    \int_{\theta_0}^{\theta_1}   \frac{1-F(\theta)}{f(\theta)} \left( v'(\theta) +F(\theta_1) /2  +  F(\theta_0)/2 \right)   \d F(\theta)  
        + \int_{\theta_1}^{\bar \theta}   \frac{1-F(\theta)}{f(\theta)} \left( v'(\theta) + s(\theta) \right)   \d F(\theta),
   \end{equation*}
   with derivative given by
   \begin{equation} \label{eqn:W'}
    W'(\theta_0) = - (1-F(\theta_0))\left( v'(\theta_0) + \frac{F(\theta_1) + F(\theta_0)}{2} \right) + \frac{f(\theta_0)}{2} \int_{\theta_0}^{\theta_1} (1-F(\theta)) \d\theta
   \end{equation}

   Finally, to prove that $W'(\theta_0)\leq 0$ under IFR, I show that $W'(\theta_0)\leq 0$ holds under a weaker condition.
   \begin{lemma} \label{lem:W'<0}
    $W'(\theta_0)\leq 0$ if  
    \begin{equation} \label{cond:W'<0}
        \frac{1}{F(\theta_1)-F(\theta_0)} \int_{\theta_0}^{\theta_1} \frac{1-F(t)}{f(t)} \d F(t) \leq \frac{1-F(\theta_0)}{f(\theta_0)}.
    \end{equation}
    The condition holds under IFR.   
    \end{lemma}
    \begin{proof}
        Rearranging the inequality yields
        \[
        \frac{f(\theta_0)}2\int_{\theta_0}^{\theta_1}(1-F(t))\d t
\leq
(1-F(\theta_0)) \frac{F(\theta_1)-F(\theta_0)}{2}.
        \]
    Substituting this into \eqref{eqn:W'} yields
\[
\begin{aligned}
W'(\theta_0)
&   \leq
-(1-F(\theta_0))
\left[
v'(\theta_0)+\frac{F(\theta_0)+F(\theta_1)}{2}
\right]
+
\frac{1-F(\theta_0)}{2}
\left[F(\theta_1)-F(\theta_0)\right] \\
& =
-(1-F(\theta_0))
\left[
v'(\theta_0)+F(\theta_0)
\right]\leq 0.
\end{aligned}
\]
The last inequality is strict for all $\theta\in(0,\bar\theta)$.

If $F$ satisfies IFR, then $\frac{1-F(t)}{f(t)}$ is decreasing, so condition~\eqref{cond:W'<0} holds.
    \end{proof}
By the lemma above, IFR implies $W'(\theta_0)\leq 0$.
Hence, decreasing $\theta_0$ increases the consumer surplus $W(\theta_0)$ in both cases.

Finally, the following counter-example shows that when IFR fails, reducing exclusion may lower consumer surplus.
\begin{example*} 
Assume \(v(\theta)=0\), \(\Theta = [1,\infty)\), and
\(
F(\theta)
=
1-
\theta^{-2}\exp\left[-\left(1-\frac1\theta\right)\right]
\).
The failure rate is
\[
\frac{f(\theta)}{1-F(\theta)}
=
\frac{2}{\theta}+\frac{1}{\theta^2},
\]
which is decreasing.

If the seller offers a single tier, revenue is
\[
R(\theta_0)
=
\frac{\theta_0}{2}
(1-F(\theta_0))(1+F(\theta_0)).
\]
Optimizing over $\theta_0$, the optimal exclusion cutoff is $\theta_0^*\approx 1.08$.

For consumer surplus
$W(\theta_0)
=
\frac{1+F(\theta_0)}{2}
\int_{\theta_0}^{\infty}(1-F(t))\d t$, we have
\[
W'(\theta_0)
=
\frac{1-F(\theta_0)}{2}
\left[
\frac{f(\theta_0)}{1-F(\theta_0)} e^{-1}\left(e^{1/\theta_0}-1\right)
-
(1+F(\theta_0))
\right].
\]
Evaluating at the seller-optimal cutoff gives $W'(\theta_0^*)\approx 0.12>0$.
Thus, reducing exclusion lowers consumer surplus.
\end{example*}
\end{proof}

\begin{proof}[Proof of Proposition~\ref{prop:welfaremax}]
First, I show that exclusion is suboptimal. Consider any feasible
mechanism $(s(\theta),p(\theta))$ with participation cutoff $\theta_0>0$. Since
$U(\theta_0)=0$, for every $\theta\geq\theta_0$,
\[
U(\theta)
=
v(\theta)-v(\theta_0)
+\int_{\theta_0}^{\theta}s(t)\d t
\]

Construct a new mechanism $(\tilde s(\theta), \tilde p(\theta))$ that sets $\tilde s(\theta)=F(\theta)$ for all $\theta<\theta_0$ and $\tilde s(\theta) =s (\theta)$ for all  $\theta\geq\theta_0$.
In words, it serves each previously excluded type at a distinct tier below all original participants, while keeping original participants' status unchanged. This is feasible because the newly served types are strictly separated below
the original participants and therefore do not affect their status.

Setting $\tilde p(0)=0$ and $\tilde p(\theta)= \int_{0}^{\theta} t \d \tilde s(t)$, we have
\[
\widetilde U(\theta)
=
v(\theta)+\int_0^\theta F(t)\d t\geq0,
\qquad \mbox{ for all } \theta<\theta_0,
\]
and, for all $\theta\geq\theta_0$,
\[
\widetilde U(\theta)
=
U(\theta)+v(\theta_0)+\int_0^{\theta_0}F(t)\d t
>
U(\theta).
\]
Thus, new participants are better off than not participating, while original participants are also strictly better off (because they pay less), so the new mechanism yields strictly higher consumer surplus. Hence, exclusion is never optimal.

    Therefore, welfare maximization under the nonnegative price constraint reduces to the following problem:
    $$\max_{s\in\MPS(F)} \int_{0}^{\bar \theta} \left( \frac{1-F(\theta)}{f(\theta)}  (s(\theta)+ v'(\theta))  + v(0)  \right)  \d F(\theta).$$
    Let $H(\theta) = \int_0^\theta \frac{1-F(t)}{f(t)} \d F(t)$ and $\conv H$ denote its convex hull in quantile space.
    By Proposition 2 in KMS, we have
    \begin{enumerate}
        \item If $\frac{1-F(\theta)}{f(\theta)}$ is decreasing (i.e., IFR), then $s^*(\theta)=1/2$.
        \item If $\frac{1-F(\theta)}{f(\theta)}$ is increasing (i.e., DFR), then $s^*(\theta) = F(\theta)$.
        \item If $\frac{1-F(\theta)}{f(\theta)}$ is not monotonic, then $s^*$ separates types when $\conv H = H$ and pools them at the same level otherwise.
    \end{enumerate}
    When the failure rate is single-peaked (single-dipped), $H= \int_0^\theta \frac{1-F(t)}{f(t)} \d F(t)$ is concave-convex (convex-concave), and the results in (iii) and (iv) follow immediately from 3.
\end{proof}

\subsection{Proofs of Section~\ref{sec:extensions}}

\begin{proof}[Proof of Observation~\ref{obs:gamma=0}]
The proof follows from the observation below:
\begin{lemma}
 If $z(\theta)$ single-crosses zero from below, the problem
\[
\max_{s\in\mathcal I(\lbar F,\bar F)}\int_{0}^{\bar\theta} z(\theta)s(\theta) \d F(\theta)
\quad (\text{subject to $s$ increasing, and }   \lbar F(\theta)\le s(\theta)\leq \bar F(\theta))
\]
has an optimal solution given by   
 $s^*(\theta) = \lbar F(\theta)$ if $z(\theta)<0$ and $s^*(\theta) = \bar F(\theta)$ if $z(\theta)>0$.
\end{lemma}
 (i) For welfare maximization under the nonnegative price constraint, $z(\theta)=\frac{1-F(\theta)}{f(\theta)}>0$; for welfare maximization without the constraint and subject to budget balance, $z(\theta)=\theta>0$.
 Thus, in either case, $s^*(\theta) = F(\theta)$ for all $\theta\in[0,\bar\theta]$.

 (ii) For revenue maximization, $z(\theta)=J(\theta)$, and the assumption $v(\theta)=\alpha\theta$ implies that $J_v(\theta)=\alpha J(\theta)$ and that the objective is $\int_{0}^{\bar\theta} J(\theta)(s(\theta)+\alpha) \d F(\theta)$.
 The result follows from the single-crossing assumption on $J(\theta)$.
\end{proof}

\begin{proof}[Proof of Observation~\ref{obs:gamma=1}]
    The proof follows from the lemma below:
    \begin{lemma}
The problem \[
    \max_{s \text{ increasing},\; \theta_0}\int_{\theta_0}^{\bar\theta} z(\theta)s(\theta) \d F(\theta)
\]
subject to
\[ s(\theta)=0 \quad\text{for all }\theta<\theta_0, \qquad F(\theta)\le s(\theta)\le 1 \quad\text{for all }\theta\in[\theta_0,\bar\theta] \]
has an optimal solution of the form $ s^*(\theta)=\mathbf 1_{\theta\ge\theta_0^*}$ for some $\theta_0^*\in[0,\bar\theta]$.
\end{lemma}
\begin{proof}
    Consider a relaxed problem
    \[\max_{s\in \mathcal M}\int_{0}^{\bar\theta} z(\theta)s(\theta) \d F(\theta),\]
    where the constraint on $s$ is relaxed to $s\in \mathcal M \equiv \{s\colon \Theta\to [0,1] \mid s \text{ increasing} \}$.
    Then, a
standard extreme-point argument implies that a posted-price mechanism $s(\theta)=\mathbf 1_{\theta\geq \theta_0^*}$ for some $\theta_0^*\in[0,\bar\theta]$ is optimal.
Because it satisfies the original constraints, it is also optimal for the original problem.
\end{proof}

(i) For welfare maximization, because $z(\theta)>0$ for all $\theta\in[0,\bar\theta]$, we have $\theta_0^*=0$.

(ii) For revenue maximization, because the maximized revenue is $\int_{\theta_0}^{\bar\theta} J(\theta) \d F(\theta)= \theta_0(1-F(\theta_0))$, we have $\theta_0^*<\bar\theta$.
\end{proof}

\begin{proof}[Proof of Proposition~\ref{prop:convex_concave}]
    I prove the result for revenue maximization (i.e., $z=J$)  only; other cases ($z=J_\lambda$ and $z=\frac{1-F(\theta)}{f(\theta)}$) can be proven analogously.

    The revenue maximization problem can be written as
    \begin{equation}
      \text{[M]}\qquad  \sup_{{\tilde s \in \overline{\conv} \mathcal F(\theta_0)},\,\theta_0} \int_{\theta_0}^{\bar \theta} J(\theta) \tilde s(\theta)  \d F(\theta) + v(\theta_0) (1-F(\theta_0)).
    \end{equation}
    For any fixed $\theta_0$, the objective is linear in $\tilde s$.
    Thus, the supremum over a set $\mathcal F(\theta_0)$ is the same as the supremum over its closed convex hull.
    Thus, the problem can be rewritten as
    \begin{equation}
       \begin{aligned}
        \text{[M]}\qquad  
        \sup_{s\in\ext\MPS(F\cdot \mathbf1_{[\theta_0,\bar\theta]}), \theta_0} \int_{\theta_0}^{\bar \theta} J(\theta) \phi(s(\theta))  \d F(\theta) + v(\theta_0) (1-F(\theta_0)).
        \end{aligned}
    \end{equation}
    Now consider a \emph{relaxed} problem [R] where the constraint is relaxed to $s \in \MPS(F\cdot \mathbf1_{[\theta_0,\bar\theta]})$.
        \begin{equation*}
        \text{[R]}\qquad \max_{s\in\MPS(F\cdot \mathbf1_{[\theta_0,\bar\theta]}), \theta_0} \int_{\theta_0}^{\bar \theta} J(\theta) \phi(s(\theta))  \d F(\theta) + v(\theta_0) (1-F(\theta_0)).  
        \end{equation*}
    The rest of the proof solves the relaxed problem [R] for each $\theta_0$ and then shows that its solution is an extreme point of $\MPS(F\cdot \mathbf1_{[\theta_0,\bar\theta]})$.
    Hence, it is also the solution to the original problem [M].

    (i) Suppose $\phi$ is strictly increasing and convex. Let $\theta_0^J =  \inf\{\theta\in[0,\bar\theta]\colon J(\theta)> 0\}$.
        Because $v(\theta)=\alpha\theta$, we have $J_v(\theta) = \alpha J(\theta)$, so the objective can be written as
        \[
        \int_{\theta_0}^{\bar \theta} J(\theta) \phi(s(\theta)) + J_v(\theta) \d F(\theta) 
        = \int_{\theta_0}^{\bar \theta} J(\theta) \left( \phi(s(\theta)) + \alpha  \right) \d F(\theta),
        \]
        where the integrand is positive if and only if  $\theta>\theta_0^J$.
        Thus, the optimal cutoff type is $\theta_0^J$.
    
    For all $\theta\geq \theta_0^J$, the kernel $K(s,\theta)=J(\theta)\phi(s)$ is convex in $s$ and supermodular in $(s,\theta)$, since
    $J(\theta)$ is positive and increasing. 
    Fix a participation cutoff $\theta_0\geq\theta_0^J$. By the Fan--Lorentz theorem, a finer
    status profile weakly increases expected revenue, so the revenue-maximizing status profile is 
    $ s^*(\theta)=F(\theta) \cdot \mathbf 1_{[\theta_0,\bar\theta]}$.

        Therefore, $s^*(\theta)=F(\theta) \cdot \mathbf 1_{[\theta_0^J,\bar\theta]}$ is the solution to the relaxed problem [R].
 Because $s^*(\theta)$ is an extreme point of $\MPS(F\cdot \mathbf1_{[\theta_0,\bar\theta]})$, it is also the solution to the original problem [M].

    (ii)  Suppose $\phi$ is strictly increasing and concave, and that
$J(\theta)\phi'(F(\theta))$ is increasing. Since $\phi'>0$,
$J(\theta)\phi'(F(\theta))$ has the same sign as $J(\theta)$.
Therefore, $J$ also single-crosses zero from below. 
For the same reason as in (i), because $v(\theta)=\alpha\theta$, the objective in the relaxed problem is
\[
\int_{\theta_0}^{\bar \theta} J(\theta) \left( \phi(s(\theta)) + \alpha  \right) \d F(\theta)
\]
and the optimal cutoff type is $\theta_0^J$.

Consider any feasible status profile
$s \in \MPS(F\cdot\mathbf{1}_{[\theta_0^J,\bar\theta]})$.
By the concavity of $\phi$,
\[
\phi(s(\theta))-\phi(F(\theta))
\leq
\phi'(F(\theta))(s(\theta)-F(\theta)).
\]
Multiplying by $J(\theta)\geq 0$ and integrating yields
\begin{align*}
\int_{\theta_0}^{\bar\theta}
J(\theta)
[\phi(s(\theta))-\phi(F(\theta))]
\d  F(\theta)\leq
\int_{\theta_0}^{\bar\theta}
J(\theta)\phi'(F(\theta))
(s(\theta)-F(\theta))
\d  F(\theta).
\end{align*}

Define
\[
D(\theta)
=
\int_{\theta}^{\bar\theta}
(F(t)-s(t))\d  F(t).
\]
Feasibility implies that $D(\theta)\geq 0$ for all
$\theta\in[\theta_0,\bar\theta]$, with
$D(\theta_0)=D(\bar\theta)=0$. %
Hence, integration by parts gives
\begin{align*}
\int_{\theta_0}^{\bar\theta}
J(\theta)\phi'(F(\theta))
(s(\theta)-F(\theta))
\d  F(\theta) =
\int_{\theta_0}^{\bar\theta}
J(\theta)\phi'(F(\theta))\d  D(\theta) =
-\int_{\theta_0}^{\bar\theta}
D(\theta)\,
\d  \left( J(\theta)\phi'(F(\theta)) \right) 
\leq 0,
\end{align*}
where the last inequality follows because
$J(\theta)\phi'(F(\theta))$ is increasing.

Therefore,
\[
\int_{\theta_0}^{\bar\theta}
J(\theta)\phi(s(\theta))\d  F(\theta)
\leq
\int_{\theta_0}^{\bar\theta}
J(\theta)\phi(F(\theta))\d  F(\theta),
\]
so $s^*=F(\theta)\cdot\mathbf{1}_{[\theta_0,\bar\theta]}$ is the optimal solution to the relaxed problem [R].
Because $s^*(\theta)$ is an extreme point of $\MPS(F\cdot \mathbf1_{[\theta_0,\bar\theta]})$, it is also the solution to the original problem [M].
    \end{proof}

\begin{proof}[Proof of Proposition~\ref{prop:intrinsic}]
    First, I show the following lemma, which implies that the optimal mechanism is induced by a deterministic increasing allocation.
\begin{lemma}[Deterministic monotone implementation] \label{lemma:compatibility}
If $J(\theta)$ is increasing, the optimal
solution is induced by a deterministic
increasing allocation $\chi:\Theta\to \tilde X$ such that
\[
    x(\theta)=\chi(\theta),\qquad
    s(\theta)=S(\chi(\theta),G_\chi).
\]
\end{lemma}
\begin{proof}[Proof of Lemma~\ref{lemma:compatibility}]
Fix any feasible randomized allocation $\chi(\theta,\omega)$. For each
realization $\omega$, let $\chi^\uparrow(\cdot,\omega)$ be the increasing
rearrangement of $\chi(\cdot,\omega)$ with respect to $F$. Then,
\(
    G_{\chi^\uparrow}(\cdot\mid \omega)
    =
    G_\chi(\cdot\mid \omega)
\).
Since $S(x,G)$ is increasing in $x$ for fixed $G$, the function $x+S(x,
G_\chi(\cdot\mid\omega))$
is increasing in $x$. Hence, because $J$ is increasing, the monotone
rearrangement inequality gives
\[
\begin{aligned}
&\int_\Theta J(\theta)
\left[
\chi^\uparrow(\theta,\omega)
+
S\!\left(\chi^\uparrow(\theta,\omega),
G_{\chi^\uparrow}(\cdot\mid\omega)\right)
\right]dF(\theta) \\
&\qquad\geq
\int_\Theta J(\theta)
\left[
\chi(\theta,\omega)
+
S\!\left(\chi(\theta,\omega),
G_{\chi}(\cdot\mid\omega)\right)
\right]dF(\theta).
\end{aligned}
\]
Moreover, since $\chi^\uparrow(\cdot,\omega)$ and $\chi(\cdot,\omega)$ have
the same distribution under $F$,
\[
    \int_\Theta c(\chi^\uparrow(\theta,\omega))\,dF(\theta)
    =
    \int_\Theta c(\chi(\theta,\omega))\,dF(\theta).
\]
Thus we can restrict attention to randomized allocations that are increasing
in $\theta$ in every realization $\omega$.

Now suppose $\chi(\cdot,\omega)$ is increasing for every $\omega$. Define the
deterministic average allocation $\bar x(\theta)=\E_\omega[\chi(\theta,\omega)]$,
which is also increasing. The quality part of revenue is unchanged because
\[
    \E_\omega\int_\Theta J(\theta)\chi(\theta,\omega)\,dF(\theta)
    =
    \int_\Theta J(\theta)\bar x(\theta)\,dF(\theta).
\]
By Jensen's inequality and the strict convexity of $c$,
\[
    \E_\omega\int_\Theta c(\chi(\theta,\omega))\,dF(\theta)
    \geq
    \int_\Theta c(\bar x(\theta))\,dF(\theta).
\]
The inequality is strict unless $\chi(\theta,\omega)$ is constant in $\omega$ for almost every $\theta$.

It remains to compare the status part. Let
\[
    s_\omega(\theta)
    =
    S\!\left(\chi(\theta,\omega),
    G_{\chi}(\cdot\mid\omega)\right),
    \qquad 
    \bar s(\theta)
    =
    S(\bar x(\theta),G_{\bar x}).
\]
Since $\bar x(\theta)=\E_\omega[\chi(\theta,\omega)]$ and each
$\chi(\cdot,\omega)$ is increasing in $\theta$, whenever $\bar x$ is flat on an interval,
each $\chi(\cdot,\omega)$ must also be flat on that interval for almost every
$\omega$. Hence, every pooling interval of $\bar x$ is contained in a pooling interval of $\chi(\cdot,\omega)$. Therefore, $\bar s \succ s_\omega$ in the majorization order.
By Lemma~\ref{lem:fanlorentz}, because $J$ is increasing, 
this implies
\[
    \int_\Theta J(\theta)\bar s(\theta)\,dF(\theta)
    \geq
    \int_\Theta J(\theta)s_\omega(\theta)\,dF(\theta)
\]
for almost every $\omega$, and therefore 
\[
    \int_\Theta J(\theta)\bar s(\theta)\,dF(\theta)
    \geq
    \E_\omega\int_\Theta J(\theta)s_\omega(\theta)\,dF(\theta)
\]
in expectation.

Combining the three parts, the deterministic increasing allocation
$\bar x$ (strictly) dominates the original (nondegenerate) randomized allocation. Finally, it is straightforward to verify that the
deterministic allocation $\bar x$ induces the interim quality $\bar x$ and status $\bar s=
    S(\bar x(\theta),G_{\bar x})$ at the same time. %
\end{proof}

By Lemma~\ref{lemma:compatibility}, it is without loss of optimality to restrict
attention to deterministic increasing allocation rules
$\chi:\Theta\to\tilde X$ such that $x(\theta)=\chi(\theta)$ and $s(\theta)=S(\chi(\theta),G_\chi)$.

Consider a relaxed optimization problem that ignores the compatibility constraint:
\[
  \max_{s\in\MPS_w(F),\; x \text{ increasing}}   \int_0^{\bar\theta}
    \left[
        J(\theta)s(\theta)
        + J(\theta)x(\theta)
        - c(x(\theta))
    \right]dF(\theta).
\]
Recall that $\theta_0^J = \inf\{\theta\in[0,\bar\theta]\colon J(\theta)> 0\}$ denotes the cutoff type such that $J(\theta)<0$ for $\theta<\theta_0^J$ and $J(\theta)>0$ for $\theta> \theta_0^J$.

Consider first the quality $x(\theta)$.
Using pointwise maximization, the solution is 
\[
    x^*(\theta)=
    \begin{cases}
        0, & \text{if } \theta<\theta_0^J, \\
        c'^{-1}(J(\theta)), & \text{if } \theta\geq\theta_0^J.
    \end{cases}
\]
When $J$ is increasing, $x^*(\theta)$ is increasing.
For all $\theta\geq \theta_0^J$, $x^*(\theta)$ is strictly increasing if $J(\theta)$ is strictly increasing, and $x^*(\theta)$ is constant if $J(\theta)$ is constant.

Now consider the status $s(\theta)$. By the same argument as in
Proposition~\ref{prop:revmax}, the revenue-maximizing status
allocation excludes types below $\theta_0^J$.
On $[\theta_0^J,\bar\theta]$, on intervals where $J(\theta)$ is strictly increasing, the optimal status is $s^*(\theta)=F(\theta)$.
On any maximal interval $[\underline\theta_i,\bar\theta_i]\subseteq
[\theta_0^J,\bar\theta]$ on which $J(\theta)=\bar J_i$ is constant, we have
\[
    \int_{\underline\theta_i}^{\bar\theta_i} J(\theta)s(\theta)\,dF(\theta)
    =
    \bar J_i
    \int_{\underline\theta_i}^{\bar\theta_i} s(\theta)\,dF(\theta),
\]
Thus, it is without loss of optimality to set
\[
    s^*(\theta)=\frac{F(\underline\theta_i)+F(\bar\theta_i)}{2},
    \qquad \text{for all }\theta\in[\underline\theta_i,\bar\theta_i].
\]

On the participating interval $[\theta_0^J,\bar\theta]$, in the regions where
$J$ is strictly increasing, the quality scheme $x^*(\theta)$ is strictly increasing and induces separating status $s^*(\theta)=F(\theta)$;
in each maximal interval $[\underline\theta_i,\bar\theta_i]\subseteq
[\theta_0^J,\bar\theta]$ on which $J(\theta)$ is constant, the quality scheme $x^*(\theta)$ is constant and induces constant status $s^*(\theta)=(F(\underline\theta_i)+F(\bar\theta_i))/2$.
Hence, in either case, we have
\[
    s^*(\theta)=S(x^*(\theta),G_{x^*})
    \qquad
    \text{for all } \theta\geq\theta_0^J,
\]
i.e., the optimal status and quality schemes are induced by the same allocation $\chi^* =x^* $.
Moreover, the exclusion cutoff $\theta_0^J$ is the same for both $s^*(\theta)$ and $x^*(\theta)$.
Hence, the solution to the relaxed problem is compatible and solves
the original problem.
\end{proof}

\begin{proof}[Proof of Proposition~\ref{prop:negstatus}]
    Let $\MPS^- (F)$ denote the lower set of $\MPS_0(F)$, that is, the set of increasing functions $s\colon \Theta\to(-\infty,1]$ that satisfy $s\leq \hat{s}$ for some $\hat{s} \in \MPS_0(F)$.
    
By the same standard arguments, it is optimal to set $U(\theta_0)=0$, and the revenue maximization problem is
\[ \max_{s \in \MPS^-(F),\theta_0}
\int_{\theta_0}^{\bar \theta}   J(\theta) s(\theta)    \d F(\theta) + v(\theta_0) (1-F(\theta_0)).
\]
For all $\theta\geq \theta_0^J=\inf\{\theta\in\Theta:J(\theta)>0\}$,  because $J(\theta)$ is increasing and nonnegative, we have $s^*(\theta) = F(\theta)$.

For all $\theta<\theta_0^J$, we have $J(\theta)<0$, so the revenue is decreasing in $s(\theta)$.
Because $U(\theta_0)=0$, to maintain $U(\theta) \geq 0$ (IR), the allocation must satisfy $U'(\theta) = s(\theta) + v'(\theta) \geq 0$. Consequently, the lowest permissible allocation is $s^*(\theta) = -v'(\theta)$.
Moreover, $p(\theta)=-\theta v'(\theta) + v(\theta)\geq0$ for all $\theta\leq \theta_0^J$ follows from the concavity of $v(\theta)$ and that $v(0)\geq0$.

Compared to the revenue-maximizing mechanism characterized in Proposition~\ref{prop:revmax}, the revenue gain is
\begin{align*}
\Delta R
&=
\int_0^{\theta_0^*}
\left[-J(\theta)v'(\theta)+J_v(\theta)\right]\d F(\theta)  +
\int_{\theta_0^*}^{\theta_0^J}
J(\theta)\left[-v'(\theta)-F(\theta)\right]\d F(\theta)\\
&=
\int_0^{\theta_0^*}
\left[v(\theta)-\theta v'(\theta)\right]\d F(\theta)
-
\int_{\theta_0^*}^{\theta_0^J}
J(\theta)\left[F(\theta)+v'(\theta)\right]\d F(\theta).
\end{align*}
By the concavity of $v$ and $v(0)\geq 0$, we have
\[
v(\theta)-\theta v'(\theta)\geq v(0)\geq 0
\qquad\text{for all }\theta\geq 0,
\]
so the first term is nonnegative. The second term is also nonnegative because $J(\theta)<0$ for all $\theta\in(\theta_0^*,\theta_0^J)$.
Moreover, if $v(\theta)-\theta v'(\theta)>0$ for all $\theta>0$, the first term is strictly positive whenever $\theta_0^*>0$. If instead
$\theta_0^*=0$, then $f(0)<\infty$ implies $J(0)<0$
and hence $\theta_0^J>0$. Since $F(\theta)>0$ for every $\theta>0$ and
$v'(\theta)\geq0$, the second term is strictly positive. Therefore, the revenue gain is strictly
positive.

For consumer surplus, types below $\theta_0^*$ obtain zero payoffs under both mechanisms.
Types in $(\theta_0^*,\theta_0^J)$ obtain strictly positive payoffs in the benchmark
but zero payoffs here. 
By Corollary~\ref{cor:exclusion}, the optimal exclusion cutoff $\theta_0^*\leq\theta_0^J$, so the consumer surplus decreases.
If $f(0)<\infty$ and $v(\theta)-\theta v'(\theta)>0$ for all $\theta>0$, we have $\theta_0^*<\theta_0^J$, so the consumer surplus strictly decreases.
\end{proof}

\begin{proof}[Proof of Proposition~\ref{prop:suffering}]
    (i) 
    Let $L(\theta) = \theta + \frac{F(\theta)}{f(\theta)}$.
    After exchanging the order of integration and setting $U(\theta_0)=0$, the revenue can be written as
    \begin{align}
        R 
        = \int_{0}^{\theta_0} \left( L(\theta) s(\theta)  + v(\theta) + \frac{F(\theta)}{f(\theta)} v'(\theta) \right) \d F(\theta)
    \end{align}
    Then, the proof is similar to that of Proposition~\ref{prop:revmax}. 
    
   For part (c), because  $s(\theta) = [F(\theta)+(1-F(\theta_0))]\cdot\mathbf1_{\theta\leq\theta_0}$ is optimal, we have 
   \[R'(\theta_0)
=
\left[
v(\theta_0)
+\frac{F(\theta_0)}{f(\theta_0)}
(v'(\theta_0)+1)
+\theta_0(1-F(\theta_0))
\right]f(\theta_0) .\]
Define
\[
\psi_L(\theta)
:=
v(\theta)
+\frac{F(\theta)}{f(\theta)}\bigl(v'(\theta)+1\bigr)
+\theta\bigl(1-F(\theta)\bigr).
\]
Because $L(\theta)=\theta+F(\theta)/f(\theta)$ is increasing,
$v'(\theta)\leq-1$, and $v''(\theta)\leq0$, we have
\[
\psi_L'(\theta)
=
L'(\theta)\bigl(v'(\theta)+1\bigr)
+\frac{F(\theta)}{f(\theta)}v''(\theta)
-F(\theta)-\theta f(\theta)
\leq0.
\]
Thus, if
\[
\psi_L(\bar\theta)
=
v(\bar\theta)
+\frac{v'(\bar\theta)+1}{f(\bar\theta)}
\geq0,
\]
we have $R'(\theta_0)\geq0$ for all $\theta_0<\bar\theta$, so zero exclusion is optimal
(i.e., $\theta_0^*=\bar\theta$).

    (ii) The proof is similar to that of Proposition~\ref{prop:approx}.
    The only difference is that because $v'(\theta)\leq0$, all types $\theta\leq \theta(p)$ buy the good.
    Thus, the auxiliary problem becomes a screening problem in procurement where the \emph{lower} types receive the good.

   (iii) Similarly to the proof of Proposition~\ref{prop:welfaremax}, for any $\theta_0<\bar\theta$, we have
   \[U(\theta)
=
v(\theta)-v(\theta_0)
-\int_\theta^{\theta_0}s(t)\d t
\leq v(\theta) \]
for all $\theta\leq\theta_0$, and thus
  \[
  W(\theta_0) = \int_{0}^{\theta_0}U(\theta)\d F(\theta) \leq \E[v(\theta)] < \E[v(\theta)+\theta/2].
  \]
  Hence, exclusion is dominated by total pooling $s(\theta)=1/2$, which results in a consumer surplus of $\E[v(\theta)+\theta/2]$.   Then, the argument is the same as that of Proposition~\ref{prop:welfaremax}.
    \end{proof}

\end{appendices}

\singlespacing
\bibliographystyle{jpe}

\bibliography{status.bib}

\end{document}